\documentclass[11pt]{article}

\usepackage{fullpage}
\usepackage{amsmath}
\usepackage{amssymb}
\usepackage[english]{babel}
\usepackage[footnotesize]{caption}
\usepackage{braket}
\usepackage[usenames,dvipsnames]{xcolor}
\usepackage{graphicx}
\usepackage[pagebackref]{hyperref}
\usepackage{ifthen}

\newcommand{\qed}{{ }\hfill$\Box$}
\newcommand{\inst}[1]{{}$^{#1}$}

\newcommand\ZZ{\mathbb{Z}}
\newcommand{\probname}{Prob}

\newcommand{\bin}{\ensuremath{\{0,1\}}}

\newcommand{\npol}{\ensuremath{\mathcal{NP}}}

\newcommand{\olrk}[1]{\ifx\nursymbol#1\else \mskip4.5mu plus 0.5mu\left(#1\right)\fi}
\newcommand{\elrk}[1]{\ifx\nursymbol#1\else \mskip4.5mu plus0.5mu\left[\mskip2.5mu plus0.5mu #1\right]\fi}

\newcommand{\negl}[1]{\ensuremath{\textrm{negl}\olrk{#1}}}
\newcommand{\prob}[1]{\ensuremath{\operatorname{\probname}\elrk{#1}}}

\newcommand{\xor}{\oplus}

\newcommand{\und}{\wedge}

\newcommand{\Angle}[1]{\ensuremath{\left\langle #1\right\rangle}}
\newcommand{\abs}[1]{\ensuremath{\left\lvert #1 \right\rvert}}
\newcommand{\norm}[1]{\ensuremath{\left\|#1\right\|}}
\newcommand{\advstyle}[1]{\mathcal{#1}}
\newcommand{\algostyle}[1]{\textsf{#1}}
\newcommand{\keystyle}[1]{\textit{#1}}
\newcommand{\schemestyle}[1]{\mathcal{#1}}

\newcommand{\adv}{\ensuremath{\advstyle{A}}}

\newcommand{\ddv}{\ensuremath{\advstyle{D}}}

\newcommand{\exec}{\ensuremath{\leftarrow}}

\newcommand{\pk}{\ensuremath{\keystyle{pk}}}

\newcommand{\sk}{\ensuremath{\keystyle{sk}}}

\newcommand{\sig}{\ensuremath{\algostyle{Sig}}}

\newcommand{\verify}{\ensuremath{\algostyle{Vf}}}
\newcommand{\sigsch}{\ensuremath{\schemestyle{S}}}

\newcommand{\A}{\mathcal{A}}

\newcommand{\X}{\mathcal{X}}
\newcommand{\Y}{\mathcal{Y}}
\newcommand{\D}{\mathcal{D}}
\newcommand{\B}{\mathcal{B}}

\newcommand{\skgen}{\ensuremath{\algostyle{SKGen}}}
\newcommand{\sign}{\sig}
\newcommand{\vf}{\ensuremath{\algostyle{Vf}}}
\newcommand{\svf}{\ensuremath{\algostyle{SVf}}}
\newcommand{\algfont}[1]{\ensuremath{\mathsf{#1}}}
\newcommand{\advfont}[1]{\ensuremath{\mathcal{#1}}}
\newcommand{\distr}{\ensuremath{\mathcal{D}}}
\newcommand{\hilbert}{\ensuremath{\mathcal{H}}}
\newcommand{\identity}{\mathbb{I}}
\newcommand{\qadv}{\ensuremath{\advfont{A}_\text{Q}}}
\newcommand{\qbdv}{\ensuremath{\advfont{B}_\text{Q}}}

\newcommand{\qext}{\ensuremath{\advfont{K}_\text{Q}}}
\newcommand{\qmeta}{\ensuremath{\advfont{M}_\text{Q}}}

\newcommand{\cL}{\ensuremath{\mathcal{L}}}
\newcommand{\cR}{\ensuremath{\mathcal{R}}}
\newcommand{\prov}{\ensuremath{\mathcal{P}}}
\newcommand{\ver}{\ensuremath{\mathcal{V}}}
\newcommand{\Sim}{\ensuremath{\mathsf{Sim}}}
\newcommand{\extractor}{\ensuremath{\mathsf{Ext}}}
\newcommand{\com}{\ensuremath{\mathsf{com}}}
\newcommand{\Com}{\ensuremath{\textsc{Com}}}
\newcommand{\Rnd}{\ensuremath{\textsc{Rnd}}}
\newcommand{\ch}{\ensuremath{\mathsf{ch}}}

\newcommand{\resp}{\ensuremath{\mathsf{rsp}}}

\newcommand{\Smplrnd}{\ensuremath{\textsc{SmplRnd}}}

\newcommand{\igen}{\ensuremath{\mathsf{Inst}}}
\newcommand{\prf}{\algfont{PRF}}
\newcommand{\secpar}{\ensuremath{\lambda}}
\newcommand{\myitem}[1][]{\item[\normalfont{\textsc{#1.}}]}

\newcommand{\wdctxt}{witness-independent commitment}
\newcommand{\wictxt}{\wdctxt}

\newcommand{\crs}{\ensuremath{\text{crs}}}
\newcommand{\poly}{\ensuremath{\text{poly}}}
\newcommand{\rand}{\ensuremath{\xleftarrow{\$}}}
\newcommand{\mysc}{\ensuremath{\textsf{SC}}}
\newcommand{\octxt}{oblivious commitments}
\newcommand{\ochead}{Oblivious Commitments}
\newcommand{\game}{\ensuremath{\textsc{Game}}}

\providecommand{\email}[1]{\texttt{#1}}
\providecommand{\mytitle}[2][\@noarg]{\title{\textbf{#2\ifthenelse{\equal{#1}{\@noarg}}{}{\\{\textbf{\large #1}}}}\\[0.5ex]}}
\providecommand{\myauthor}[1]{\author{#1}}
\providecommand{\myinstitute}[1]{\date{\small\renewcommand{\and}{\\}#1}}

\providecommand{\institute}[1]{#1}
\providecommand{\myinstitute}[1]{\institute{}}

\providecommand{\mysection}[1]{\section{#1}}
\providecommand{\mysubsection}[1]{\subsection{#1}}
\providecommand{\myparagraph}[1]{\paragraph{#1}}

\renewcommand*{\backref}[1]{}
\renewcommand*{\backrefalt}[4]{\ifcase #1 \or \footnotesize (Cited on page~#2.) \else \footnotesize (Cited on pages~#2.) \fi}
\renewcommand{\H}{\mathcal{H}}
\renewcommand{\O}{H}

\newtheorem{theorem}{Theorem}[section]
\newtheorem{lemma}[theorem]{Lemma}
\newtheorem{definition}[theorem]{Def\/inition}

\newtheorem{construction}[theorem]{Construction}

\newenvironment{proof}[1][\@noarg]{ \begin{trivlist} \providecommand{\proofname}{Proof} \item[] \ifthenelse{\equal{#1}{\@noarg}}{\textit{\proofname{}. }}{\textit{\proofname{} (#1). }}}{ \end{trivlist}}

\renewenvironment{abstract}{\begin{quote}\begin{small}\textbf{\abstractname.}}{\vspace*{1.5ex}\end{small}\end{quote}}

\begin{document}

\title{The Fiat--Shamir Transformation in a Quantum World}

\date{\vspace*{-5ex}}

\myauthor{\"Ozg\"ur Dagdelen \and Marc Fischlin \and Tommaso Gagliardoni}

\myinstitute{\inst{} Technische Universit\"at Darmstadt, Germany\\   \email{www.cryptoplexity.de} \\
    \email{oezguer.dagdelen$\,@\,$cased.de}\qquad \email{marc.fischlin$\,@\,$gmail.com}\qquad  \email{tommaso$\,@\,$gagliardoni.net} }
    
\maketitle

\begin{abstract}
The Fiat-Shamir transformation is a famous technique to turn identification schemes into signature
schemes. The derived scheme is provably secure in the random-oracle model against classical adversaries.
Still, the technique has also been suggested to be used in connection with quantum-immune identification
schemes, in order to get quantum-immune signature schemes. However, a recent paper by Boneh
et al. (Asiacrypt 2011) has raised the issue that results in the random-oracle model may
not be immediately applicable to quantum adversaries, because such adversaries should be allowed
to query the random oracle in superposition. It has been unclear if the Fiat-Shamir technique
is still secure in this quantum oracle model (QROM).

Here, we discuss that giving proofs for the Fiat-Shamir transformation in the QROM is presumably hard. We show that
there cannot be black-box extractors, as long as the underlying quantum-immune identification
scheme is secure against active adversaries and the first message of the prover is independent
of its witness. Most schemes are of this type. We then discuss that for some schemes
one may be able to resurrect the Fiat-Shamir result in the QROM by modifying the underlying protocol first.
We discuss in particular
a version of the Lyubashevsky scheme which is provably secure in the QROM.
\end{abstract}

\mysection{Introduction}

The Fiat-Shamir transformation \cite{C:FiaSha86} is a well-known method to remove interaction
in three-move identification schemes between a prover and verifier, by letting
the verifier's challenge $\ch$ be determined via a hash function $H$ applied to the prover's first message $\com$.
Currently, the only generic, provably secure instantiation is by modeling the hash function $H$
as a random oracle \cite{CCS:BelRog93,JC:PoiSte00}. In general,
finding secure instantiations based on \emph{standard}
hash functions is hard for some schemes, as shown in \cite{FOCS:GolKal03,TCCpre:BDS13}. However,
these negative results usually rely on peculiar identification schemes, such that for specific schemes, especially
more practical ones, such instantiations may still be possible.

\myparagraph{The Quantum Random-Oracle model}
Recently, the Fiat-Shamir transformation has also been applied to schemes which are
advertised as being based on quantum-immune primitives, e.g.,
\cite{AC:Lyubashevsky09,EPRINT:BarMis10,AC:GorKatVai10,PROVSEC:CLRS10,SAC:CayVerAla10,C:SakShiHiw11,MGS11,PKC:Sakumoto12,CHES:GnuLyuPop12,EC:AFLT12,SCN:CamNevRuc12,EC:AJLTVW12}. Interestingly, the
proofs for such schemes still investigate classical adversaries only. It seems
unclear if (and how) one can transfer the proofs to the quantum case. Besides the problem
that the classical Fiat-Shamir proof \cite{JC:PoiSte00} relies on rewinding the adversary, which is often
considered to be critical for quantum adversaries (albeit not impossible \cite{STOC:Watrous06,EC:Unruh12}),
a bigger discomfort seems to lie in the usage of the random-oracle model in presence
of quantum adversaries.

As pointed out by Boneh et al.~\cite{AC:BDFLSZ11} the minimal requirement for random oracles
in the quantum world should be \emph{quantum access}. Since the random oracle is
eventually replaced by a standard hash function, a quantum adversary could evaluate this
hash function in superposition, while still ignoring any advanced attacks exploiting the structure
of the actual hash function. To reflect this in the random-oracle model, \cite{AC:BDFLSZ11}
argue that the quantum adversary should be also allowed to query the random oracle in
superposition. That is, the adversary should be able to query the oracle on a
state $\ket{\varphi}=\sum_{x} \alpha_x\ket{x}\ket{0}$ and in return would get
$\sum_{x}\alpha_x\ket{x}\ket{H(x)}$. This model is called the quantum random-oracle
model (QROM).

Boneh et al.~\cite{AC:BDFLSZ11} discuss some classical constructions for encryption and
signatures which remain secure in the QROM. They do not cover Fiat-Shamir signatures, though.
Subsequently, Boneh and Zhandry \cite{C:Zhandry12,FOCS:zhandry,EPRINT:BonZha12b} investigate
further primitives with quantum access, such as pseudorandom functions and MACs. Still, the
question about the security of the Fiat-Shamir transform in the QROM raised in~\cite{AC:BDFLSZ11} remained open.

\myparagraph{Fiat-Shamir Transform in the QROM}
Here, we give evidence that conducting security proofs for Fiat-Shamir transformed schemes and black-box 
adversaries is hard, thus yielding a negative result about the provable security of such schemes. More specifically, 
we use the meta-reduction technique to rule out the existence of quantum extractors 
with black-box access to a quantum adversary against the converted (classical) scheme. If such extractors would exist 
then the meta-reduction, together with the extractor, yields a quantum
algorithm which breaks the active security of the identification scheme. 
Our result covers \emph{any} identification scheme, as long as the
prover's initial commitment in the scheme is independent of the witness, 
and if the scheme itself is secure against active quantum attacks where a malicious verifier may first
interact with the genuine prover before trying to impersonate or, as we only demand here, to compute
a witness afterwards. Albeit not quantum-immune,
the classical schemes of Schnorr \cite{C:Schnorr89}, Guillou and Quisquater \cite{C:GuiQui88}, and Feige, 
Fiat and Shamir~\cite{JC:FeiFiaSha88} are conceivably of this type (see also~\cite{C:BelPal02}).
Quantum-immune candidates are, for instance,  \cite{C:MicVad03,PKC:Lyubashevsky08,AC:KawTanXag08,MGS11,C:SakShiHiw11,EC:AJLTVW12}.

Our negative result does not primarily rely on the rewinding problem for quantum adversaries;
our extractor may rewind the adversary (in a black-box way).
Instead, our result is rather based on the adversary's possibility to hide actual queries to the quantum random oracle
in a ``superposition cloud'', such that the extractor or simulator cannot elicit or implant necessary information
for such queries. In fact, our result reveals a technical subtlety in the QROM which previous 
works \cite{AC:BDFLSZ11,FOCS:zhandry,C:Zhandry12,EPRINT:BonZha12b}
have not addressed at all, or at most implicitly. It refers to the question how a simulator or extractor can answer superposition
queries $\sum_{x} \alpha_x\ket{x}\ket{0}$.

A possible option is to allow the simulator to reply with an arbitrary quantum state $\ket{\psi}=\sum_x \beta_x \ket{x}\ket{y_x}$, e.g.,
by swapping the state from its local registers to the ancilla bits for the answer in order to make this step unitary. This seems to
somehow generalize the classical situation where the simulator on input $x$ returns an arbitrary string $y$ for $H(x)$. Yet,
the main difference is that returning an arbitrary state $\ket{\psi}$ could also be used to eliminate some of the
input values $x$, i.e., by setting $\beta_x=0$. This is more than what the simulator is able to do in the classical setting,
where the adversary can uniquely identify the preimage~$x$ to the answer. In the extreme the simulator in the quantum case,
upon receiving a (quantum version of) a classical state $\ket{x}\ket{0}$, could simply
reply with an (arbitrary) quantum state $\ket{\psi}$. Since quantum states are in general indistinguishable, in contrast
to the classical case the adversary here would potentially continue its execution for inputs which it has not queried for.

In previous works \cite{AC:BDFLSZ11,C:Zhandry12,FOCS:zhandry,EPRINT:BonZha12b} the
simulator specifies a classical (possibly probabilistic) function $h$
which maps the adversary query $\sum_{x} \alpha_x\ket{x}\ket{0}$ to the reply
$\sum_{x} \alpha_x\ket{x}\ket{h(x)}$. Note that the function $h$ is not given explicitly to the adversary,
and that it can thus implement keyed functions like a pseudorandom function (as in \cite{AC:BDFLSZ11}).
This basically allows the simulator to freely assign values $h(x)$ to each string $x$, without being able to change
the input values. It also corresponds to the idea that, if the random oracle is eventually replaced by
an actual hash function, the quantum adversary can check that the hash function is classical,
even if the adversary does not aim to exploit any structural weaknesses (such that we still hide $h$ from the adversary).

We thus adopt the approach of letting the simulator determine the quantum answer via a classical probabilistic function $h$.
In fact, our impossibility hinges on this property but which we believe to be rather ``natural'' for the
aforementioned reasons. From a mere technical point of view it at least clearly
identifies possible venues to bypass our hardness result. In our case we allow the simulator to specify the (efficient)
function $h$ adaptively for each query, still covering techniques like programmability in the classical setting.
Albeit this is sometimes considered
to be a doubtful property \cite{AC:FLRSST10} this strengthens our impossibility result in this regard.

\myparagraph{Positive Results}
We conclude with some positive result. It remains open if one can ``rescue'' plain Fiat-Shamir for schemes
which are not actively secure, or to prove that alternative but still reasonably efficient approaches work.
However, we can show that the Fiat-Shamir technique in general \emph{does}
provide a secure signature scheme in the QROM if the protocol allows for \octxt. Roughly,
this means that the honest verifier generates the prover's first message $\com$ obliviously by sampling a random string
and sends $\com$ to the prover. In the random oracle transformed scheme the commitment is thus computed via the
random oracle, together with the challenge.
Such schemes are usually not actively secure against malicious verifiers. Nonetheless, we stress that in order to derive a secure signature
scheme via the Fiat-Shamir transform, the underlying identification scheme merely needs to provide passive security 
and honest-verifier zero-knowledge.

To make the above transformation work, we need that the prover is able to compute the response for commitments chosen 
obliviously to the prover. For some schemes this is indeed possible if the prover holds some trapdoor information.
Albeit not quantum-immune, it is instructive to look at the Guillou-Quisquater RSA-based proof of knowledge \cite{C:GuiQui88}
where the prover shows knowledge~of $w\in\ZZ_N^*$ with $w^e={y}\bmod{N}$ for $x=(e,N,y)$. For an oblivious commitment 
the prover would need to compute an $e$-th root for a given commitment $R\in\ZZ_N^*$. If the witness would contain the prime
factorization of $N$, instead of the $e$-th root of $y$, this would indeed be possible.
As a concrete example we discuss that we
can still devise a provably secure signature version of Lyubashevsky's identification scheme \cite{EC:Lyubashevsky12} via our method.
Before, Lyubashevsky only showed security in the classical random-oracle model, despite using an allegedly quantum-immune
primitive.

Our results are summarized in Figure~\ref{fig:results}. Actively secure identification schemes with witness-independent
commitments (lower right area) are hard to prove secure in the quantum random oracle model. Schemes with oblivious and 
therefore witness-independent commitments can be proven secure (upper right area). Schemes outside of this area may be patched according to our idea
exemplified for Lyubashevsky's scheme to turn them into secure signature schemes in the QROM. For any other identification scheme
the question remains open.

\begin{figure}[t]
\begin{center}
\includegraphics[width=0.7\textwidth]{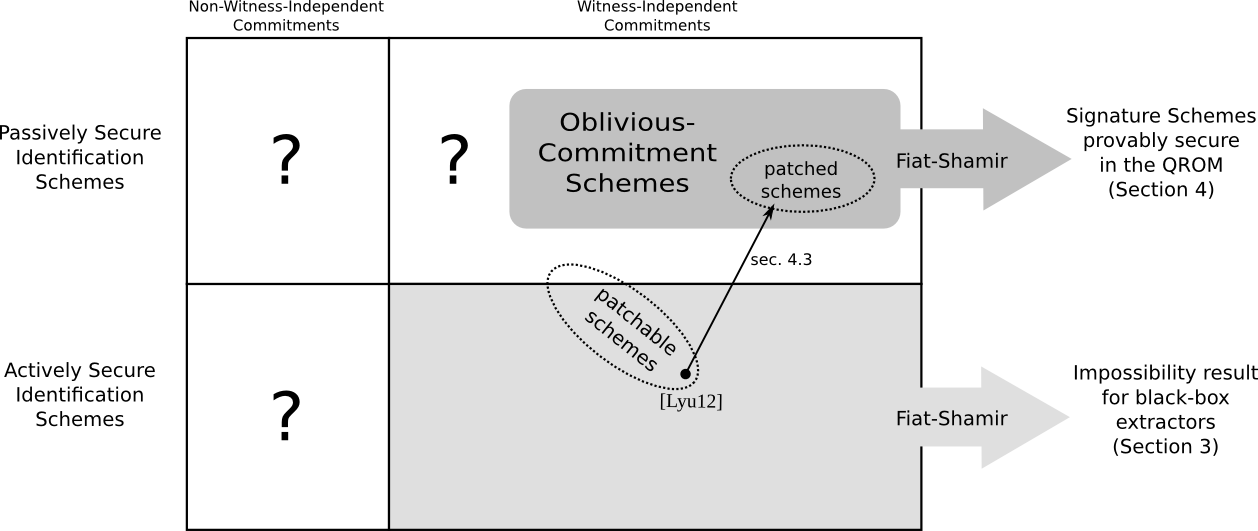}
\end{center}
\caption{Possibility and impossibility results for the Fiat-Shamir transform of identification schemes in the QROM.}
\label{fig:results}
\end{figure}

\myparagraph{Related work} Since the introduction of the quantum-accessible random-oracle model~\cite{AC:BDFLSZ11}, 
several works propose cryptographic primitives or revisit their security against quantum algorithms in this stronger 
model~\cite{FOCS:zhandry,C:Zhandry12,EPRINT:BonZha12b}. In~\cite{EPRINT:DFNS11}, Damg\aa{}rd et al.~look at the security 
of cryptographic protocols where the underlying primitives or even parties can be queried by an adversary in a superposition. We here
investigate the scenario in which the quantum adversary can only interact classically with the classical honest parties, except
for the locally evaluable random oracle.

In a concurrent and independent work, Boneh and Zhandry~\cite{EPRINT:BonZha13} analyze the security of signature schemes under
quantum chosen-message attacks, i.e., the adversary in the unforgeability notion of the signature scheme may query the signing
oracle in superposition and, eventually, in the quantum random oracle model. Our negative result carries over to the 
quantum chosen-message attack model as well, since our impossibility holds even allowing only classical queries to the 
signing oracle. Moreover, while the authors of \cite{EPRINT:BonZha13} show how to obtain signature schemes secure in the 
quantum-accessible signing oracle model, starting with schemes
secure in the classical sense, we focus on signature schemes and proofs of knowledge derived from identification schemes via
the Fiat-Shamir paradigm.

\mysection{Preliminaries}

We first describe (to the level we require it) quantum computations and then recall
the quantum random-oracle model of Boneh et al.~\cite{AC:BDFLSZ11}. We also introduce
the notion of $\Sigma$-protocols to which the Fiat-Shamir transformation applies.
In the end of this section, we recall the definition of signature schemes and its security.

\mysubsection{Quantum Computations in the QROM}

We first briefly recall facts about quantum computations and set some notation; for more details, we refer to \cite{ChaNie00}.
Our description follows \cite{AC:BDFLSZ11} closely.

\myparagraph{Quantum Systems}
A quantum system $A$ is associated to a complex Hilbert space $\hilbert_A$ of finite dimension and with an inner
product $\Angle{\cdot | \cdot}$. The state of the system is given by a (class of) normalized vector $\ket{\varphi}\in\hilbert_A$ with
Euclidean norm
$\|\ket{\varphi}\|=\sqrt{\Angle{\varphi | \varphi}}=1$.

The joint or composite quantum state of two quantum systems $A$ and $B$ over spaces $\hilbert_A$ and $\hilbert_B$, respectively,
is given through the tensor product $\hilbert_A\otimes\hilbert_B$.  The product state of $\ket{\varphi_A}\in\hilbert_A$ and 
$\ket{\varphi_B}\in\hilbert_B$ is denoted by $\ket{\varphi_A}\otimes\ket{\varphi_B}$. We sometimes simply write
$\ket{\varphi_A}\ket{\varphi_B}$ or $\ket{\varphi_A,\varphi_B}$.  An $n$-qubit system is associated in the joint quantum 
system of $n$ two-dimensional Hilbert spaces.
The standard orthonormal computational basis $\ket{x}$ for such a system is given by $\ket{x}=\ket{x_1}\otimes\dots\otimes\ket{x_n}$ 
for $x=x_1\dots x_n\in\bin^n$. We often assume that any (classical) bit string $x$ is encoded into a quantum state as $\ket{x}$,
and vice versa we sometimes view such a state simply as a classical state.
Any pure $n$-qubit state $\ket{\varphi}$ can be expressed as a superposition in the computational basis as 
$\ket{\varphi} = \sum_{x \in \bin^n} \alpha_x \ket{x}$ where $\alpha_x$ are complex amplitudes 
obeying $\sum_{x \in \bin^n} |\alpha_x|^2 =1$.

\myparagraph{Quantum Computations}
Evolutions of quantum systems are described by unitary transformations with $\identity_A$ being the identity transformation on
register $A$. For a composite quantum system over $\hilbert_A\otimes\hilbert_B$ and a transformation $U_A$ acting only on $\hilbert_A$,
it is understood that $U_A\ket{\varphi_A}\ket{\varphi_B}$ is a simplification of $(U_A\otimes\identity_B)\ket{\varphi_A}\ket{\varphi_B}$.
Note that any unitary operation and, thus, any quantum operation, is invertible.

Information can be extracted from a quantum state $\ket{\varphi}$ by performing a positive-operator valued measurement (POVM) $M=\{M_i\}_i$ 
with positive semi-definite measurement operators $M_i$ that sum to the identity $\sum_i M_i = \identity$. Outcome $i$ is obtained with 
probability $p_i = \bra{\varphi} M_i\ket{\varphi}$.
A special case are projective measurements such as the measurement in the computational basis of the state 
$\ket{\varphi}=\sum_x \alpha_x \ket{x}$ which yields outcome $x$ with probability $|\alpha_x|^2$. Measurements can refer to a subset of 
quantum registers and are in general not invertible.

We model a quantum algorithm $\qadv$ with access to oracles $O_1,O_2,\dots$ by a sequence of unitary transformations
\[ U_1, O_1, U_2,\dots, O_{T-1},U_T \] over $m=\text{poly}(n)$ qubits. Here, oracle function $O_i:\bin^a\to\bin^b$ maps the final $a+b$
qubits  from basis state $\ket{x}\ket{y}$ to $\ket{x}\ket{y\xor O_i(x)}$ for $x\in\bin^a$ and $y \in \bin^b$. This mapping is inverse to
itself.
We can let the oracles share (secret) state by reserving some qubits for the $O_i$'s only, on which the $U_j$'s cannot operate.
Note that the algorithm $\qadv$ may also receive some (quantum) input $\ket{\psi}$. The adversary may also perform measurements.
We sometimes
write $\qadv^{\ket{O_1(\cdot)},\ket{O_2(\cdot)},\dots}(\ket{\psi})$ for the output.

To introduce asymptotics we assume that $\qadv$ is actually a sequence of such transformation sequences, indexed by parameter $n$, 
and that each transformation sequence is composed out of quantum systems for input, output, oracle calls, and work space (of sufficiently 
many qubits). To measure polynomial running time, we assume that each $U_i$ is approximated (to sufficient precision) by members of a set 
of universal gates (say, Hadamard, phase, CNOT and $\pi/8$; for sake of
concreteness \cite{ChaNie00}), where at most polynomially many gates are used. Furthermore, $T=T(n)$ is assumed to be polynomial, too.

\myparagraph{Quantum Random Oracles}
We can now define the quantum random-oracle model by picking a random function $H$ for a given domain and range, and letting
(a subset of) the oracles $O_i$ evaluate $H$ on the input in superposition, namely those $O_i$'s which correspond to hash
oracle queries. In this case the quantum adversary can evaluate the hash function in parallel for many inputs
by querying the oracle about $\sum_x \alpha_x \ket{x}$ and obtaining $\sum_x \alpha_x \ket{H(x)}$, appropriately encoded
as described above.
Note that the output distribution $\qadv^{\ket{O_1(\cdot)},\ket{O_2(\cdot)},\dots}(\ket{\psi})$ now refers to the
$\qadv$'s measurements and the choice of $H$ (and the random choices for the other oracles, if existing).

\mysubsection{Classical Interactive Proofs of Knowledge}

Here, we review the basic definition of $\Sigma$-protocols and show the classical Fiat-Shamir transformation which 
converts the interactive $\Sigma$-protocols into non-interactive proof of knowledge (PoK) protocols (in the random-oracle model).

Let $\cL\in\npol$ be a language with a (polynomially computable) relation $\cR$, i.e.,\ $x\in\cL$ if and only if there exists some
$w\in\bin^*$ such that $\cR(x,w) = 1$ and $\abs{w} = poly(\abs{x})$ for any $x$. As usual, $w$ is called a witness for 
$x\in\cL$ (and $x$ is sometimes called a ``theorem'' or statement). We sometimes use the notation $\cR_\secpar$ to denote the 
set of pairs $(x,w)$ in $\cR$ of some complexity related to the security parameter,
e.g., if $|x|=\secpar$.

\myparagraph{$\Sigma$-Protocols}
The well-known class of $\Sigma$-protocols between a prover $\prov$ and a verifier $\ver$ allows $\prov$ to convince $\ver$ that 
it knows a witness $w$ for a public theorem $x\in\cL$, without giving $\ver$ non-trivially computable information beyond this fact.
Informally, a $\Sigma$-protocol consists of three messages $(\com,\ch,\resp)$ where the first message $\com$ is sent by
$\prov$ and the challenge $\ch$ is sampled uniformly from a challenge space by the verifier. We write
$(\com,\ch,\resp)\exec\Angle{\prov(x,w),\ver(x)}$ for the randomized output of an interaction between $\prov$ and $\ver$.
We denote individual messages of the (stateful) prover in such an execution by
$\com\exec\prov(x,w)$ and $\resp\exec\prov(x,w,\com,\ch)$, respectively. Analogously, we denote the verifier's steps
by $\ch\exec\ver(x,\com)$ and $d\exec\ver(x,\com,\ch,\resp)$ for the challenge step and the final decision.

\begin{definition}[$\Sigma$-Protocol]
A $\Sigma$-protocol $(\prov,\ver)$ for an $\npol$-relation $\cR$ satisfies the following properties:
\begin{description}

\myitem[Completeness]
For any security parameter \secpar, any $(x,w)\in\cR_\secpar$, any $(\com,\ch,\resp)\leftarrow\Angle{\prov(x,w),\ver(x)}$ it holds $\ver(x,\com,\ch,\resp)=1$.

\myitem[Public-Coin] For any security parameter \secpar, any $(x,w)\in\cR_\secpar$, and any $\com\leftarrow\prov(x,w)$, the challenge 
$\ch\leftarrow\ver(x,\com)$ is uniform on $\{0,1\}^{\ell(\secpar)}$ where $\ell$ is some polynomial function.

\myitem[Special Soundness] Given $(\com,\ch,\resp)$ and $(\com,\ch',\resp')$ for $x\in\cL$ (with $\ch\neq \ch'$) 
where \linebreak  $\ver(x,\com,\ch,\resp)= \ver(x,\com,\ch',\resp')=1$, there exists a PPT algorithm $\extractor$ (the extractor) 
which for any such input
outputs a witness $w\leftarrow \extractor(x,\com,\ch,\resp,\ch',\resp')$ for $x$ satisfying $\cR(x,w) = 1$.

\myitem[Honest-Verifier Zero-Knowledge (HVZK)] There exists a PPT algorithm $\Sim$ (the zero-knowledge simulator) which, on 
input $x\in\cL$, outputs a transcript $(\com,\ch,\resp)$ that is computationally indistinguishable from a valid transcript derived 
in a $\prov$-$\ver$ interaction. That is, for any polynomial-time quantum algorithm $\ddv=(\ddv_0,\ddv_1)$ the following distributions 
are indistinguishable:

\begin{itemize}
  \item Let $(x,w,\mathsf{state})\leftarrow \ddv_0(1^\secpar)$. If $\cR(x,w)=1$, then $(\com,\ch,\resp)\leftarrow \langle\prov(x,w),\ver(x)\rangle$; \\ else, $(\com,\ch,\resp)\leftarrow \bot$. Output $\ddv_1(\com,\ch,\resp,\mathsf{state})$.
  \item Let $(x,w,\mathsf{state})\leftarrow \ddv_0(1^\secpar)$. If $\cR(x,w)=1$, then $(\com,\ch,\resp)\leftarrow \Sim(x)$; \\ else, $(\com,\ch,\resp)\leftarrow \bot$. Output $\ddv_1(\com,\ch,\resp,\mathsf{state})$.
\end{itemize}
Here, $\mathsf{state}$ can be a quantum state.

\end{description}
\end{definition}

\myparagraph{Fiat-Shamir (FS) Transformation}
The Fiat-Shamir transformation of a $\Sigma$-protocol $(\prov,\ver)$ is the same protocol but
where the computation of $\ch$ is done as $\ch\exec H(x,\com)$ instead of $\exec\ver(x,\com)$. Here, $H$ is a public
hash function which is usually modeled as a random oracle, in which case we speak of the Fiat-Shamir transformation of $(\prov,\ver)$ 
in the random-oracle model. Note that we include $x$ in the hash computation, but all of our results remain valid if $x$ is omitted 
from the input. If applying the FS transformation to a (passively-secure) identification protocol one obtains a signature scheme, if 
the hash computation also includes the message $m$ to be signed.

\mysubsection{Quantum Extractors and the FS Transform}

\myparagraph{Quantum Extractors in the QROM}
Next, we describe a black-box quantum extractor. Roughly, this extractor should be able to output a witness $w$
for a statement $x$ given black-box access to the adversarial prover. There are different possibilities
to define this notion, e.g., see the discussion in \cite{EC:Unruh12}. Here, we take a simple approach
which is geared towards the application of the FS transform to build secure signature schemes.
Namely, we assume that, if a quantum adversary $\qadv$ on input $x$
and with access to a quantum-accessible random oracle
has a non-negligible probability of outputting a valid proof $(\com,\ch,\resp)$, then there is an extractor $\qext$
which on input $x$ and with black-box access to $\qadv$ outputs a valid witness with non-negligible probability,
too.

We need to specify how the extractor simulates the quantum-accessible random oracle.
This time we view the extractor $\qext$ as a sequence of unitary transformations $U_1,U_2,U_3,\dots$,
interleaved with interactions with the adversary $\qadv$, now represented
as the sequence of (stateful) oracles $O_1,O_2,\dots$ to which $\qext$ has access to. Here each $O_i$ corresponds to the local
computations of the adversary until the ``next interaction with the outside world''. In our case this will be
basically the hash queries $\ket{\varphi}$ to the quantum-accessible random oracle. We stipulate $\qext$ to write
the (circuit description of a) classical function $h$ with the expected input/output length, and which we assume for the moment
to be deterministic, in some register before making the next call to an oracle. Before this call is
then actually made, the hash function $h$ is first applied to the
quantum state $\ket{\varphi}=\sum_x \alpha_x\ket{x}\ket{0}$ of the previous oracle in the sense that the next oracle is called
with $\sum_x \alpha_x\ket{x}\ket{h(x)}$. Note that we can enforce this behavior formally  by restricting $\qext$'s steps $U_1,U_2,\dots$
to be of this described form above.

At some point the adversary will return some classical proof $(\com,\ch,\resp)$ for $x$.
To allow the extractor to rewind the adversary we assume that
the extractor can  invoke another run
with the adversary (for the same randomness, or possibly fresh randomness, appropriately encoded in
the behavior of oracles). If the reduction asks to keep the same randomness then since the adversary
only receives classical input $x$, this corresponds to a reset to
the initial state. 
Since we do not consider adversaries with auxiliary quantum input, but only with classical input,
such resets are admissible.

For our negative result we assume that the adversary does not perform any measurements before eventually creating the final
output, whereas our positive result also works if the adversary measures in between.
This is not a restriction, since in the meta-reduction technique we are allowed to choose a specific adversary,
 without having to consider more general cases.
Note that the intrinsic ``quantum randomness'' of the adversary is fresh for each rewound run
but, for our negative result, since measurements of the adversary are postponed till the end, the extractor can
re-create the same quantum state as before at every interaction point.
Also note that the extractor can measure any quantum query of the adversary to the random oracle
but then cannot continue the simulation of this instance (unless the adversary chose a classical query in the first place).
The latter reflects the fact that the extractor cannot change the quantum input state for answering the adversary's queries to
the random oracle.

In summary, the black-box extractor can: (a) run several instances of the adversary
from the start for the same or fresh classical randomness, possibly reaching the same quantum state as in previous
executions when the adversary interacts with external oracles, (b) for each query to the QRO either measure and abort this execution,
or provide a hash function $h$, and (c) observe the adversary's final output. The black-box extractor cannot, for instance, interfere
with the adversary's program and postpone or perform additional measurements, nor rewind the adversary between interactions with the outside
world, nor tamper with the internal state of the adversary.
As a consequence, the extractor cannot observe the adversary's queries, 
but we still allow the extractor to access queries if these are classical. In particular, the extractor may choose $h$ adaptively
but not based on quantum queries (only on classical queries).
We motivate this model with the observation that, in meaningful scenarios, the extractor should only be able to give a classical description 
of $h$, which is then ``quantum-implemented'' by the adversary $\qadv$ through a ``quantum programmable oracle gate''; the gate itself will 
be part of the adversary's circuit, and hence will be outside the extractor's influence. Purification of the adversary is also not allowed, 
since this would discard those adversaries which perform measurements, and would hence hinder the notion of black-box access.

For an interesting security notion computing a witness from $x$ only should be infeasible, even for a quantum adversary.
To this end we assume that there is an efficient instance generator $\igen$ which on input~$1^\secpar$ outputs a pair
$(x,w)\in\cR$ such that any polynomial-time quantum algorithm on (classical) input~$x$ returns some classical string
$w'$ with $(x,w')\in\cR$, is negligible (over the random choices of $\igen$ and the quantum algorithm).
We say $\igen$ is a \emph{hard instance generator for relation~$\cR$}.

\begin{definition}[Black-Box Extractor for $\Sigma$-Protocol in the QROM]\label{def:bbext}
Let $(\prov,\ver)$  be a $\Sigma$-pro\-tocol for an $\npol$-relation $\cR$ with hard instance generator $\igen$.
Then a black-box extractor $\qext$ is a polynomial-time quantum algorithm (as above) such that
for any quantum adversary $\qadv$ with quantum access to oracle $H$, it holds that, if
\[ \prob{\ver^H(x,\com,\ch,\resp)=1\text{ for }(x,w)\exec\igen(1^\secpar);(\com,\ch,\resp)\exec\qadv^{\ket{H}}(x)}\not\approx 0 \]
is not negligible, then
\[ \prob{(x,w')\in\cR\text{ for }(x,w)\exec\igen(1^\secpar);w'\exec\qext^{\qadv}(x)}\not\approx 0 \]
is also not negligible.
\end{definition}
For our negative (and our positive) results we look at special cases of black-box extractors, denoted

\emph{input-respecting} extractors. This

means that the extractor only runs the adversary on the given input $x$.
All known extractors are of this kind, and in general it is unclear how to take advantage of executions for different $x'$.

\myparagraph{On Probabilistic Hash Functions}
We note that we could also allow the extractor to output a description of a \emph{probabilistic} hash function $h$
to answer each random oracle call. This means that, when evaluated for some string $x$, the reply is $y=h(x;r)$
for some randomness $r$ (which is outside of the extractor's control). In this sense a query
$\ket{\varphi}=\sum_x\alpha_x\ket{x}\ket{0}$ in superposition returns $\ket{\varphi}=\sum_x\alpha_x\ket{x}\ket{h(x;r_x)}$
for independently chosen $r_x$ for each $x$.

We can reduce the case of probabilistic functions $h$ to deterministic ones, if we assume quantum-accessible
pseudorandom functions \cite{AC:BDFLSZ11}. These functions are indistinguishable from random functions for quantum adversaries,
even if  queried in superposition. In our setting, in the deterministic case the extractor incorporates
the description of the pseudorandom function for a randomly chosen key $\kappa$ into the description of the
deterministic hash function, $h'(x)=h(x;\prf_\kappa(x))$. Since the hash function description is not presented to the
adversary, using such derandomized hash functions cannot decrease the extractor's success probability
significantly. This argument can be carried out formally by a reduction to the quantum-accessible pseudorandom
function, i.e., by forwarding each query $\ket{\varphi}$ of the QROM adversary to the random or pseudorandom function oracle,
and evaluating $h$ as before on $x$ and the oracle's reply. Using a general technique in \cite{C:Zhandry12} we can even
replace the assumption about the pseudorandom function and use a $q$-wise independent function instead.

\mysubsection{Signature Schemes and Their Security}\label{sec:sig}

Here, we recall the definition of signature schemes and their security.

\begin{definition}[Signature Scheme]\label{def:sig}
A (digital) signature scheme (in the random-oracle model) consists of three efficient algorithms (\skgen, \sign, \svf) defined as follows.
\begin{description}
	\myitem[Key Generation] On input the security parameter $1^\secpar$, the probabilistic algorithm $\skgen^H$ with oracle access to
	 $H$ outputs a key pair
	$(\sk,\pk)$ where $\sk$ (resp.~$\pk$) denotes the signing key (resp.~public verification key).	
	\myitem[Signing] On input a signing key $\sk$ and a message $m$, the probabilistic algorithm $\sign^H$ outputs a signature $\sigma$.
	\myitem[Verification] On input the verification key $\pk$, a message $m$, and a signature $\sigma$, the deterministic
	algorithm $\svf^H$ outputs either $1$ ($=$ valid) or $0$ ($=$ invalid).
\end{description}
We require correctness of the verification, i.e., the verifier will always accept genuine signatures. More formally, for any security 
parameter \secpar, any $(\sk,\pk)\leftarrow\skgen(1^\secpar)$, for any message $m$, any signature $\sigma\leftarrow\sign(\sk,m)$, we have
$\svf(\pk,m,\sigma)=1$.
\end{definition}
From a signature scheme we require that no outsider should be able to forge signatures. Formally, this property is called
unforgeability against adaptively chosen-message attacks (unf-cma) and is defined as follows.
\begin{definition}[UNF-CMA Security]\label{def:sig-unf}
A (digital) signature scheme $\sigsch=(\skgen,\sig,\svf)$ in the random-oracle model is $(t,Q,\varepsilon)$-unforgeable against adaptively 
chosen-message attacks with $Q=(q_H,q_S)$ if for any algorithm \adv\ with runtime $t$ and making at most $q_H$ (resp. $q_S$) queries to 
the random oracle (resp. its signing oracle), the probability that the following experiment returns $1$ is at most $\varepsilon$.
\begin{center}
	\begin{tabbing}
123\=123\=123\=123\=123\=\kill
\> pick random function $H$\\
\> $(\sk,\pk)\rand \skgen^H(1^\secpar)$ \\
\> $(m^*,\sigma^*)\rand \adv^{H,\sign^H(\sk,\cdot)}(\pk)$\\
\> Return $1$ iff $\svf^H(\pk,m^*,\sigma^*)=1$ and $m^*\notin \mathsf{M}$.\\
\> \> Here, $\mathsf{M}$ is the set of message queried to $\sign^H(\sk,\cdot)$.
\end{tabbing}
\end{center}
The probability is taken over all coin tosses of $\skgen$, $\sign$, and \adv, and the choice of $H$.
\end{definition}
We call a signature scheme \emph{existentially unforgeable under chosen message attacks} in the (quantum) random-oracle model if for any 
PPT (quantum) algorithm making at most polynomial number of (superposition) queries to the (quantum) random oracle and classical queries 
to the signature scheme, the probability for the above experiment is negligible in the security parameter.

\mysection{Impossibility Result for Quantum-Fiat-Shamir}\label{sec:imp}

We use meta-reductions techniques to show that, if the Fiat-Shamir transformation applied to the
identification protocol would support a knowledge extractor, then we would obtain a contradiction to
the active security. That is, we first build an all-powerful quantum adversary $\qadv$ successfully generating accepted
proofs. Coming up with such an adversary is necessary to ensure that a black-box extractor $\qext$ exists in the first
place; Definition~\ref{def:bbext} only requires $\qext$ to succeed \emph{if} there is some successful adversary $\qadv$.
The adversary $\qadv$ uses its unbounded power to find a witness $w$ to its input $x$, and then uses the quantum access
to the random oracle model to ``hide'' its actual query in a superposition. The former ensures that that our adversary
is trivially able to construct a valid proof by emulating the prover for $w$, the latter prevents the extractor to
apply the rewinding techniques of  Pointcheval and Stern \cite{JC:PoiSte00} in the classical setting.
Once we have designed our adversary $\qadv$ and ensured the existence of $\qext$, 
we wrap $\qext$ into a reduction $\qmeta$ which takes the role of $\qadv$ and breaks active security.
The (quantum) meta-reduction now plays against the honest prover of the identification scheme ``on the outside'',
using the extractor ``on the inside''. In this inner interaction $\qmeta$ needs to emulate our all-powerful adversary
$\qadv$ towards the extractor, but this needs to be done efficiently in order to make sure that the meta-reduction
(with its inner interactions) is efficient.

In the argument below we assume that the extractor is 
input-respecting (i.e., forwards $x$ faithfully to the adversary).
In this case we can easily derandomize the adversary (with respect to classical randomness)
by ``hardwiring'' a key of a random function into it, which it initially applies to its input
$x$ to recover the same classical randomness for each run. Since the extractor has to work for all
adversaries, it in particular needs to succeed for those where we pick the function randomly but fix it from thereon.

\mysubsection{Assessment}

Before we dive into the technical details of our result let us re-assess the strength and weaknesses of our impossibility result:

\begin{enumerate}
\item The extractor has to choose a classical hash function $h$ for answering QRO queries. While this may be considered a restriction
 in general interactive quantum proofs, it seems to be inevitable in the QROM; it is rather a consequence of the approach where a quantum
 adversary mounts attacks in a classical setting. After all, both the honest parties as well as the adversary expect a classical hash
 function. The adversary is able to check this property easily, even if it treats the hash function
 otherwise as a black box (and may thus not be able to spot that the hash function uses (pseudo)randomness).
 We remark again that
 this approach also complies with previous efforts \cite{AC:BDFLSZ11,C:Zhandry12,FOCS:zhandry,EPRINT:BonZha12b} and our
 positive result here to answer such hash queries.

\item The extractor \emph{can} rewind the quantum adversary to any point before the final measurement.
 Recall that for our impossibility result we assume, to the advantage of the extractor, that the adversary
  does not perform any measurement until the very end.
 Since the extractor can re-run the adversary from scratch for the same classical randomness, and the ``no-cloning restriction''
  does not apply to our adversary
  with classical input, the extractor can therefore easily put the adversary in the same (quantum) state as in a previous execution,
  up to  the final measurement. However, because we consider \emph{black-box} extractors, the extractor can only
  influence
  the adversary's behavior via the answers it provides to $\qadv$'s external communication. In this sense,
  the extractor may always rewind the adversary to such communication points. We also allow the extractor to
  measure and abort at such communication points.

\item The extraction strategy by Pointcheval and Stern \cite{JC:PoiSte00} in the purely classical case \emph{can} be cast in our black-box
 extractor framework. For this the extractor would run the adversary for the same classical randomness twice, providing a lazy-sampling
 based hash function description, with different replies in the $i$-th answers in the two runs. The extractor then extracts the witness from
  two  valid signatures. This shows that a different approach than in the classical setting is necessary for extractors in the QROM.
\end{enumerate}

\mysubsection{Prerequisites}

\myparagraph{Witness-Independent Commitments}
We first identify a special subclass of $\Sigma$-protocols which our result relies upon:

\begin{definition}[$\Sigma$-protocols with \wictxt]
A $\Sigma$-protocol has \emph{wit\-ness-independent commitments} if the prover's commitment $\com$ does not depend on the witness $w$.
That is, we assume that there is a PPT algorithm~$\Com$ which, on input $x$ and some randomness~$r$, produces the same
distribution as the prover's first message for input $(x,w)$.
\end{definition}
Examples of such $\Sigma$-protocols are the well known graph-isomorphism proof \cite{C:GolMicWig86},
the Schnorr proof of knowledge \cite{JC:Schnorr91}, or the recent protocol for lattices used in an anonymous credential
system~\cite{SCN:CamNevRuc12}.
A typical example of non-\wdctxt{} $\Sigma$-protocol is the graph $3$-coloring ZKPoK scheme \cite{C:GolMicWig86}
where the prover commits to a random permutation of the coloring.

We note that perfectly hiding commitments do not suffice for our negative result. We need to be able to generate
(the superposition of) all commitments without knowledge of the witness.

\myparagraph{Weak Security Against Active Quantum Adversaries}
We next describe the underlying security of (non-transformed) $\Sigma$-protocols against a weak form of active attacks where the
adversary may use quantum power but needs to eventually compute a witness. That is, we let $\qadv^{\prov(x,w)}(x)$ be a quantum
adversary which can interact classically with several prover instances. The prover
instances can be invoked in sequential order, each time the prover starts by computing a fresh commitment $\com\leftarrow\prov(x,w)$,
and upon receiving a challenge $\ch\in\bin^\ell$ it computes the response~$\resp$. Only if it has returned this response $\prov$
can be invoked on a new session again. We say that the adversary \emph{succeeds in an active attack}
if it eventually returns some $w'$ such that
$(x,w')\in\cR$.

For an interesting security notion computing a witness from $x$ only should be infeasible, even for a quantum adversary.
To this end we assume that there is an efficient instance generator $\igen$ which on input~$1^\secpar$ outputs a pair
$(x,w)\in\cR$ such that any polynomial-time quantum algorithm on (classical) input~$x$ returns some classical string
$w'$ with $(x,w')\in\cR$, is negligible (over the random choices of $\igen$ and the quantum algorithm).
We say $\igen$ is a \emph{hard instance generator for relation~$\cR$}.

\begin{definition}[Weakly Secure $\Sigma$-Protocol Against Active Quantum Adversaries]
A $\Sigma$-protocol $(\prov,\ver)$ for an $\npol$-relation $\cR$ with hard instance generator $\igen$
is weakly secure against active quantum adversaries if for any polynomial-time quantum adversaries $\qadv$
the probability that $\qadv^{\prov(x,w)}(x)$ succeeds in an active attack for $(x,w)\exec\igen(1^\secpar)$ is
negligible (as a function of $\secpar$).
\end{definition}
We call this property weak security because it demands the adversary to compute a witness $w'$, instead
of passing only an impersonation attempt. If the adversary finds such a witness, then completeness
of the scheme implies that it can successfully impersonate. In this sense we put more restrictions on the adversary
and, thus, weaken the security guarantees.

\begin{center}
\begin{figure}[t]
\begin{center}
\includegraphics[width=5in]{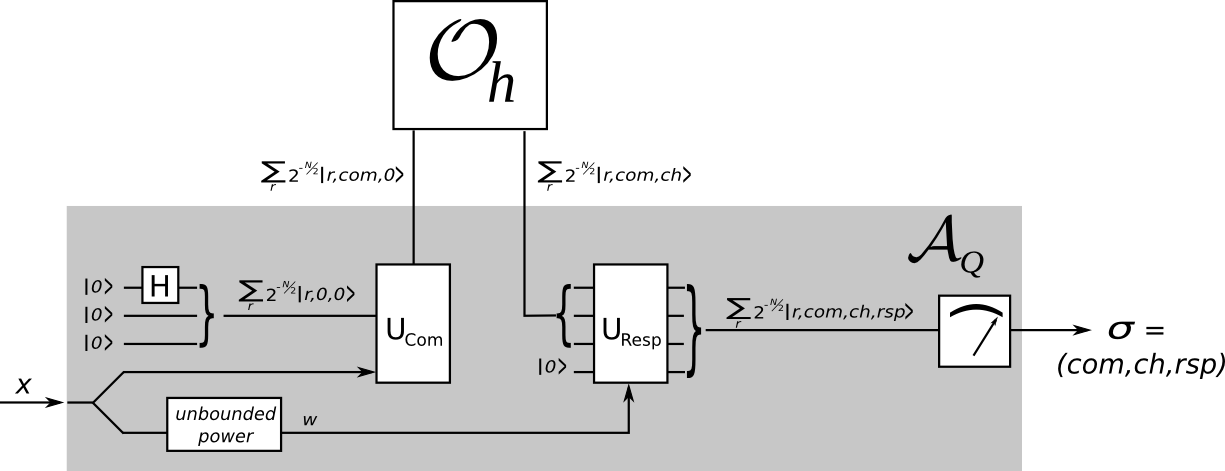}
\end{center}
\caption{The canonical adversary}
\label{fig:adv}
\end{figure}
\end{center}

\mysubsection{The Adversary and the Meta-Reduction}

\myparagraph{Adversary}
Our (unbounded) adversary works roughly as follows (see Figure~\ref{fig:adv}). It receives as input a value $x$ and first uses
its unbounded computational power to compute a random witness $w$ (according to uniform distributions of
coin tosses $\omega$ subject to $\igen(1^n;\omega)=(x,w)$, but where $\omega$ is a random function of~$x$).
Then it prepares all possible random strings $r\in\bin^{N}$ (where $N = \poly(n)$)
for the prover's algorithm in superposition. It then evaluates (a unitary version of) the classical function $\Com ()$ for computing the prover's commitment
on this superposition (and on $x$) to get a superposition of all $\ket{r}\ket{\com_{x,r}}$.
It evaluates the random oracle $H$ on the $\com$-part, i.e., to be precise, the hash values are stored
in ancilla bits such that the result is a superposition of states $\ket{r}\ket{\com_{x,r}}\ket{H(x,\com_{x,r})}$.
The adversary computes, in superposition, responses for all values and finally measures
in the computational basis, yielding a sample $(r,\com_{x,r},\ch, \resp_{x,w,r})$ for $\ch=H(x,\com_{x,r})$
where $r$ is uniform over all random strings; 
it outputs the transcript $(\com,\ch,\resp)$.

\begin{center}
\begin{figure}[t]
\centering
\includegraphics[width=5in]{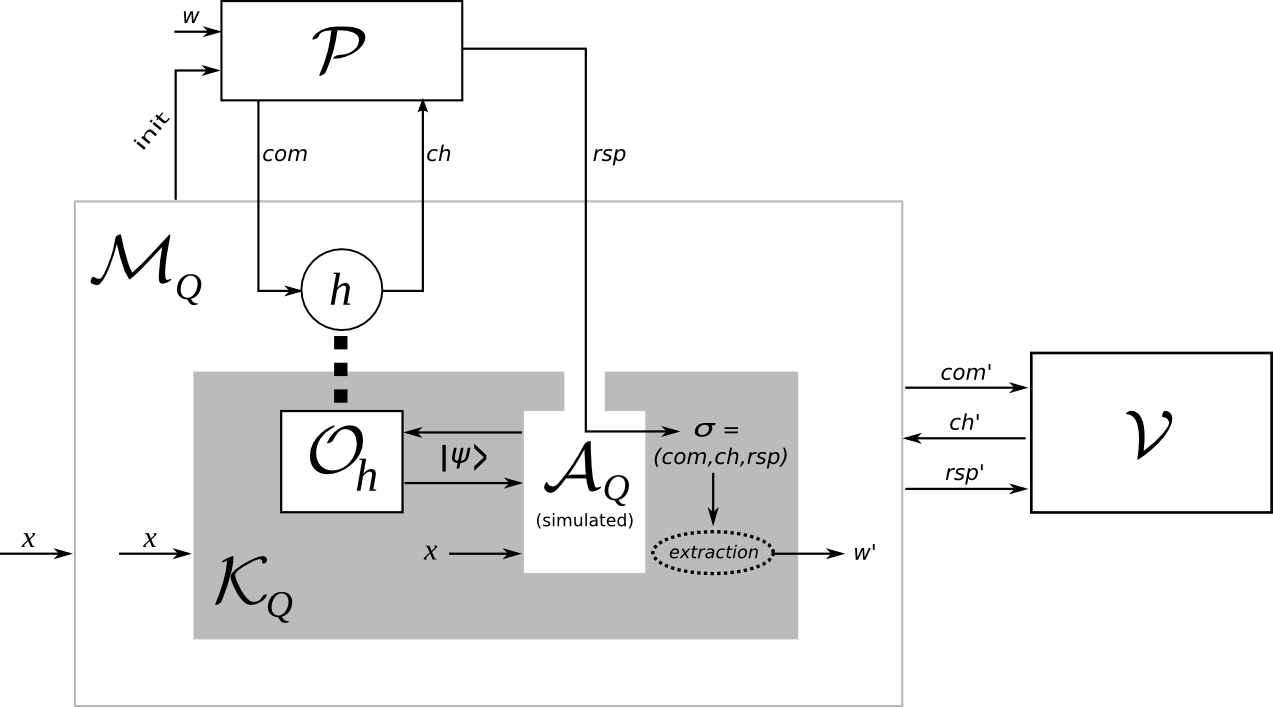}
\caption{An overview of our meta-reduction}
\label{fig:red}
\end{figure}
\end{center}

\vspace{-3.5em}

\myparagraph{The Meta-Reduction}
We illustrate the meta-reduction in Figure~\ref{fig:red}. Assume that there exists a (quantum) black-box 
extractor $\qext$ which on input~$x$, sampled according to $\igen$,
and which is also given to $\qadv$, is able to extract a witness $w$ to $x$ by running several resetting executions
of $\qadv$, each time answering $\qadv$'s (only) random oracle query $\ket{\varphi}$ by supplying a classical, possibly probabilistic
function $h$. We then build a (quantum) meta-reduction $\qmeta$ which breaks the weak security of the identification scheme
in an active attack when communicating with the classical prover.

The quantum meta-reduction $\qmeta$ receives as input the public statement $x$. It forwards it
to $\qext$ and waits until $\qext$ invokes $\qadv(x)$, which is now simulated by $\qmeta$.
For each (reset) execution the meta-reduction skips the step where the adversary would compute
the witness, and instead immediately computes the same superposition query $\ket{r}\ket{\com_{x,r}}$
as $\qadv$ and outputs it to $\qext$. When $\qext$
creates (a description of) the possibly probabilistic function $h$ we let $\qmeta$ initiate an interaction with the prover
to receive a classical sample $\com_{x,r}$, on which it evaluates $h$ to get a challenge $\ch$.
Note that $\qmeta$ in principle does not need a description of $h$ for this, but only a possibility to
compute $h$ once. The meta-reduction forwards the challenge to the prover to get a response $\resp$. It 
outputs $(\com,\ch,\resp)$ to the reduction. If the reduction eventually outputs a potential witness $w'$ then
$\qmeta$ uses this value $w'$ to break the weak security.

\mysubsection{Analysis}

For the analysis note that the extractor's perspective in each execution is identical in both cases, when interacting
with the actual adversary $\qadv$, or when interacting with the meta-reduction $\qmeta$. The reason is that
the commitments are witness-independent such that the adversary (using its computational power to first compute a witness)
and the meta-reduction computing the commitments without knowledge of a witness, create the same distribution on the
query to the random oracle.
Since up to this point the extractor's view is identical in both runs, its distribution on $h$ is also the same
in both cases. But then the quantum adversary internally computes, in superposition over all 
possible random strings $r$, the challenge $\ch\exec h(x,\com_{x,r})$ and the response $\resp_{x,w,r}$ for $x,w$, and $\ch$.
It then measures $r$ in the computational basis, such that the state collapses to a classical tuple $(\com_{x,r},\ch,\resp_{x,w,r})$
over uniformly distributed $r$.
Analogously, the meta-reduction, upon receiving $h$ (with the same distribution as in $\qadv$'s attack),
receives from the prover a commitment $\com_{x,r}$ for a uniformly distributed $r$. It then computes $\ch\exec h(x,\com_{x,r})$
and obtains $\resp_{x,w,r}$ from the prover, which is determined by $x,w,r$ and $\ch$. It returns $(\com_{x,r},\ch,\resp_{x,w,r})$
for such a uniform $r$.

In other words, $\qmeta$ considers only a single classical execution (with $r$ sampled at the outset), whereas
$\qadv$ basically first runs everything in superposition and only samples $r$ at the very end. Since all the other
computations in between are classical, the final results are identically distributed.
Furthermore, since the extractor is input-respecting, the meta-reduction can indeed answer all runs for the very
same $x$ with the help of the external prover (which only works for $x$). Analogously, the fact that the adversary always chooses,
and uses, the same witness $w$ in all runs, implies that the meta-reduction can again rely on the external prover with the
single witness $w$.

Since the all-powerful adversary succeeds with probability $1$ in the original experiment, to output a valid proof
given $x$ and access to a quantum random oracle only, the extractor must also succeed with non-negligible probability
in extracting a witness. Hence, $\qmeta$, too, succeeds with non-negligible probability in an active attack against weak security.
Furthermore, since $\qext$ runs in polynomial time, $\qmeta$ invokes at most a polynomial number of interactions
with the external prover. Altogether, we thus obtain the following theorem:

\begin{theorem}[Impossibility Result]
For any $\Sigma$-protocol $(\prov,\ver)$ with witness-independent commitments, and which is
 weakly secure against active quantum adversaries, there does not exist
an input-preserving  black-box quantum knowledge extractor for $(\prov,\ver)$.
\end{theorem}
We note that our impossibility result is cast in terms of proofs of knowledge, but can be easily adapted
for the case of signatures. In fact, the adversary $\qadv$ would be able to compute a valid proof (i.e.,
a signature) for any given message $m$ which it receives as additional input to $x$.

\myparagraph{Our Meta-Reduction and Classical Queries to the Random Oracle}
One might ask why the meta-reduction does not apply to the Fiat-Shamir transform when adversaries have only classical access to the 
random oracle. The reason is the following: if the adversary made a classical query about a single commitment (and so would the 
meta-reduction), then one could apply the rewinding technique of Pointcheval and Stern~\cite{JC:PoiSte00} changing the random oracle 
answers, and extract the underlying witness via special soundness of the identification scheme. The quantum adversary here, however, 
queries the random oracle in a superposition. In this scenario, as we explained above, the extractor is not allowed to ``read'' the 
query of the adversary unless it makes the adversary stop. In other words, the extractor cannot measure the query and then keep 
running the adversary until a valid witness is output. This intrinsic property of black-box quantum extractors, hence, makes 
``quantum'' rewinding impossible. Note that rewinding in the classical sense ---as described by Pointcheval and 
Stern~\cite{JC:PoiSte00}--- is still possible, as this essentially means to start the adversary with the same random coins. One may 
argue that it might be possible to measure the query state without disturbing $\qadv$'s behavior significantly, but as we already 
pointed out, this would lead to a non-black-box approach ---vastly more powerful than the classical read-only access.

\mysubsection{On the Necessity of Active Security}\label{sec:nas}

We briefly discuss that active security is basically necessary for an impossibility result as above. That is,
we outline a three-move protocol for any $\npol$ language which, when applying the FS transformation supports a
straight-line extractor, and is honest-verifier zero-knowledge, but  not actively secure. This holds as long as
there are quantum-immune dense encryption, and quantum-immune non-interactive zero-knowledge proofs.
The latter are classical non-interactive zero-knowledge proofs (in the common random string model) for which
simulated and genuine proofs are indistinguishable, even for \emph{quantum} distinguishers. The former are encryption
schemes which are IND-CPA against quantum adversaries (see, for example, \cite{AC:BDFLSZ11}) but where, in addition,
honestly generated public keys are quantum-indistinguishable from random strings.

The construction is based on the (classical) non-interactive zero-knowledge proofs of knowledge of
De Santis and Persiano \cite{FOCS:DeSPer92} and works as follows: The first message is irrelevant, 
e.g., we let the prover simply send 
the constant $0$ (potentially padded with redundant randomness). In the second message the verifier sends a
random string which the prover interprets as a public key $\pk$ of the dense encryption scheme and a common
random string $\crs$ for the NIZK. The prover encrypts the witness under $\pk$ and gives a NIZK that the
encrypted value forms a valid witness for the public value $x$. The verifier only checks the NIZK proof.

The protocol is clearly not secure against active (classical) adversaries because such an adversary
can create a public key $\pk$ via the key generation algorithm, thus, knowing the secret key and allowing
the adversary to recover the witness from a proof by the prover. It is, however, honest-verifier zero-knowledge
against quantum distinguishers because of the IND-CPA security and the simulatability of the NIZK hide the witness
and allow for a simulation. 
We omit a more formal argument here, as it will be covered as a special case from our general result in the
next section.

\section{Positive Results for Quantum-Fiat-Shamir}

In Section~\ref{sec:nas} we have sketched a generic construction of a $\Sigma$-protocol based on NIZKPoKs 
\cite{FOCS:DeSPer92} which can be converted to a secure NIZK-PoK against quantum adversaries in the QROM via the Fiat-Shamir 
(FS) paradigm. While the construction is rather inefficient and relies on additional
primitives and assumptions, it shows the path to a rather efficient solution: drop the requirement on active security and let the
(honest) verifier choose the commitment obliviously, i.e., such that it does not know the pre-image, together with the challenge.
If the prover is able to use a trapdoor to compute the c

\subsection{$\Sigma$-protocols with \ochead}

The following definition captures the notion of $\Sigma$-protocols  with \octxt\ formally.

\begin{definition}[$\Sigma$-protocols with \ochead]\label{def:vcc}
A $\Sigma$-protocol $(\prov,\ver)$ has  \emph{\octxt}
if there are PPT algorithms $\Com$ and $\Smplrnd$ 
such that for any $(x,w)\in\cR$ the following distributions are statistically close:
\begin{itemize}
\item Let $\com=\Com(x;\rho)$ for $\rho\exec\bin^\secpar$, $\ch\exec\ver(x,\com)$, and $\resp\exec\prov(x,w,\com,\ch)$. \\ Output $(x,w,\rho,\com,\ch,\resp)$.
\item Let $(x,w,\rho,\com,\ch,\resp)$ be a transcript of a protocol run between $\prov(x,w)$ and $\ver(x)$, \\ where $\rho\exec\Smplrnd(x,\com)$.
\end{itemize}
\end{definition}

Note that the prover is able to compute a response from the given commitment $\com$ without knowing the randomness used to compute the commitment. This is usually
achieved by placing some extra trapdoor into the witness $w$. For example, for the Guillou-Quisquater RSA based proof of knowledge \cite{C:GuiQui88}
where the prover shows knowledge of $w\in\ZZ_N^*$ with $w^e={y}\bmod{N}$ for $x=(e,N,y)$,
the prover would need to compute an $e$-th root for a given commitment $R\in\ZZ_N^*$. If the witness would contain the prime factorization 
of $N$, instead of the $e$-th root of $y$, this would indeed be possible.

$\Sigma$-protocols with oblivious commitments allow to move the generation of the commitment from the prover
to the honest verifier. For most schemes this infringes with active security, because a malicious verifier could
generate the commitment ``non-obliviously''. However, the scheme remains honest-verifier zero-knowledge, and this suffices
for deriving secure signature schemes. In particular, using random oracles one can hash into commitments by computing
the random output of the hash function and running $\Com(x;\rho)$ on this random string $\rho$ to sample a commitment obliviously.

In the sequel we therefore often identify $\rho$ with $\Com(x;\rho)$ in the sense that we assume that the hash function
maps to $\Com(x;\rho)$ directly. The existence of $\Smplrnd$ guarantees that we could ``bend'' this value back to the actual
pre-image $\rho$. In fact, for our positive result it would suffice that the distributions are computationally indistinguishable
for random $(x,w)\exec\igen(1^n)$ against quantum distinguishers.

\subsection{FS Transformation for $\Sigma$-protocols with \ochead}

We explain the FS transformation for schemes with \octxt\ for signatures only;
the case of (simulation-sound) NIZK-PoKs is similar, the difference is that for signatures the message is included
in the hash computation for signature schemes. For sake of concreteness let us give the full description of the
transformed signature scheme. We note that for the transformation we also include a random string~$r$ in the
hash computation (chosen by the signer).
Jumping ahead, we note that this source of entropy ensures simulatability of signatures; for classical $\Sigma$-protocols
this is usually given by the entropy of the initial commitment but which has been moved to the verifier here.
Recall from the previous section that we simply assume that we can hash into commitments directly, instead
of going through the mapping via $\Com$ and $\Smplrnd$.

\begin{construction}\label{constr:sig}
Let $(\prov,\ver)$ be a $\Sigma$-protocol for relation $\cR$ with \octxt\ and instance generator $\igen$.
Then construct the following signature scheme $\mathcal{S}=(\skgen, \sign, \svf)$ in the (quantum) random-oracle model:
\begin{description}
\myitem[Key Generation] $\skgen(1^\secpar)$ runs $(x,w)\exec\igen(1^\secpar)$ and returns $\sk=(x,w)$ and $\pk=x$.
\myitem[Signing] For message $m\in\bin^*$ the signing algorithm $\sign^H$ on input $\sk$, picks random \linebreak $r\rand\Rnd$ from some
superpolynomial space, computes $(\com,\ch)=H(\pk,m,r)$, and obtains $\resp\exec\prov(\pk,\sk,\com,\ch)$. The output is the signature $\sigma=(r,\com,\ch,\resp)$.
\myitem[Verification] On input $\pk$,$m$, and $\sigma=(r,\com,\ch,\resp)$ the verification algorithm $\verify^{\,H}$ outputs $1$ iff 
$\ver(\pk,\com,\ch,\resp)=1$ and $(\com,\ch)=H(\pk,m,r)$; else, it returns $0$.
\end{description}
\end{construction}
Note that one can shorten the signature size by simply outputting $\sigma=(r,\resp)$. The remaining components $(\com,\ch)$ are obtained by 
hashing the tuple $(pk,m,r)$. Next, we give the main result of this section saying that the Fiat-Shamir transform on $\Sigma$-protocols 
with \octxt\ yield a quantum-secure signature scheme.

\begin{theorem}\label{th:fs}
If $\igen$ is a hard instance generator for the relation $\cR$ and the $\Sigma$-protocol $(\prov,\ver)$ has \octxt, then the signature 
scheme in Construction~\ref{constr:sig} is existentially unforgeable under chosen message attacks against quantum adversaries in the quantum-accessible random-oracle model.
\end{theorem}
The idea is roughly as follows. Assume for the moment that we are only interested in key-only attacks and
would like to extract the secret key from an adversary $\qadv$ against the signature scheme. For given $x$ we first run the 
honest-verifier zero-knowledge simulator of the $\Sigma$-protocol to create a transcript $(\com^\star,\ch^\star,\resp^\star)$. 
We choose another random challenge $\ch'\exec\bin^{\ell}$.
Then, we run the adversary, injecting $(\com^\star,\ch')$ into the hash replies. This appropriate insertion
will be based on techniques developed by Zhandry~\cite{C:Zhandry12} to make sure that superposition queries
to the random oracle are harmless. With sufficiently large probability the adversary will then output
a proof $(\com^\star,\ch',\resp')$ from which we can, together with $(\com^\star,\ch^\star,\resp^\star)$ extract a witness due 
to the special-soundness property. Note that, if this extraction fails because the transcript $(\com^\star,\ch^\star,\resp^\star)$ 
is only simulated, we could distinguish simulated signatures from genuine ones.
We can extend this argument to chosen-message attacks by simulating signatures as in the classical case.
This is the step where we take advantage of the extra random string $r$ in order to make sure that
the previous adversary's quantum hash queries have a negligible amplitude in this value $(x,m,r)$.
Using techniques from \cite{BBBV97} we can show that changing the oracle in this case does not change
the adversary's success probability significantly.

\mysubsection{Technical Results for the Security Proof}\label{app:preresult}

We start by recalling two results from Bernstein and Vazirani~\cite{BernsteinV97} and Bennett et al.~\cite{BBBV97} which we make use 
of in the proof of Theorem~\ref{th:fs}. Before so, we introduce distance measures.

\myparagraph{Distance Measures}
For two quantum states $\ket{\varphi}=\sum\alpha_x\ket{x}$ and $\ket{\psi}=\sum\beta_x\ket{x}$ in superposition in the basis 
states $\ket{x}$, the Euclidean distance is given by $\big(\sum_x\abs{\alpha_x-\beta_x}^2\big)^{1/2}$. The total variation 
distance (aka.~statistical difference) of two distributions $\distr_0,\distr_1$
is defined through $\sum_x\abs{\prob{\distr_0=x}-\prob{\distr_1=x}}$. The following fact from \cite{BernsteinV97} upperbounds 
the total variance distance in terms of the Euclidean distance:
\begin{lemma}[{\cite[Lemma 3.6]{BernsteinV97}}]\label{th:distances}
  Let $\ket{\varphi},\ket{\psi}$ be quantum states with Euclidean distance at most $\epsilon$. Then, performing the same 
  measurement  on $\ket{\varphi},\ket{\psi}$ yields distributions with statistical  distance at most~$4\epsilon$.
\end{lemma}
Let $q_\rho(\ket{\phi_t})$ be the magnitude squared of $\rho$ in the superposition of query $t$ which we call the query 
probability of $r$ in query $t$. If we sum over all queries $t$, we get an upper bound on the total query probability of 
$r$.
The following is a result from Bennett et al~\cite{BBBV97}.
\begin{lemma}[{\cite[Theorem 3.3]{BBBV97}}]\label{th:queryprob}
  Let $\qadv$ be a quantum algorithm running in time $T$ with oracle access to $\O$. Let $\epsilon > 0$ and let 
  $S\subseteq [1,T] \times \{0,1\}^n$ be a set of time-string pairs such that $\sum_{(t,\rho)\in S}q_\rho(\ket{\phi_t})\le \epsilon$. 
  If we modify $\O$ into an oracle $\O'$ which answers each query $\rho$ at time $t$ by providing the same string $R$ (which 
  has been sampled independently form $\O$), then the Euclidean distance between the final states of $\qadv$ when invoking 
  $\O$ and $\O'$ is at most $\sqrt{T\epsilon}$.
\end{lemma}

\myparagraph{Injecting Values into Oracles}
Let us now introduce some definitions and results including so-called semi-constant distributions $\mysc_\delta$ introduced by Zhandry~\cite{C:Zhandry12}.
\begin{definition}[Semi-Constant Distributions]
Let $\H_{\X \times \Y} = \set{H:\X \rightarrow \Y}$ be a family of functions for sets $\X$ and $\Y$ and let $\delta \in [0,1]$. 
We define the \emph{semi-constant distribution} $\mysc_\delta$ as the distribution over $\H_{\X \times \Y}$ resulting from the 
following process:
\begin{itemize}
\item first, pick a random element $y \in \Y$;
\item then, for each $x \in \X$ do the following:
	\begin{itemize}
	\item with probability $\delta$, set $H(x) = y$;
	\item otherwise, set $H(x)$ to be a (uniformly) randomly chosen element in $\Y$.
	\end{itemize}
\end{itemize}
\end{definition}
Notice that $\mysc_0$ is the uniform distribution, while $\mysc_1$ is a constant distribution.
Also note that the distribution, when used within an oracle, is consistent in the sense
that the settings are chosen once at the outset. We will use this definition to describe a quantum
random oracle which has been ``reprogrammed'' on a fraction $\delta$ of its possible inputs.

The following lemma by Zhandry~\cite{C:Zhandry12} gives an upper bound on the probability that a quantum algorithm's output behavior 
changes when switching from a truly random oracle to an oracle drawn from $\mysc_\delta$ in terms of statistical distance:
\begin{lemma}[{\cite[Corollary 4.3]{C:Zhandry12}}]\label{lem:zhandrycustom}
Let $\qadv^{\ket{\O}}$ be a quantum algorithm making at most $q_H$ queries to the quantum-accessible random oracle $\O$. Let $\delta \in (0,1)$ and
let $\O^\prime$ be the oracle obtained by reprogramming $\O$ on a fraction $\delta$ of its possible inputs, i.e.,
let $\O^\prime$ be described  by distribution $\mysc_\delta$. Then,
\[
\abs{ \qadv^{\ket{\O}} - \qadv^{\ket{\O^\prime}}} \leq \frac{8}{3}\cdot q_H^4  \delta^2\; .
\]
\end{lemma}

Recall our quantum adversary $\qadv$ against the unforgeability property of the signature scheme from 
Construction~\ref{constr:sig}. It works by performing at most $q_H = \poly(\secpar)$ queries to the quantum-accessible
random oracle. This means that the statistical distance in the two cases, and in particular the
probability $\epsilon^\prime$ that $\qadv^{\ket{\O^\prime}}$ successfully forges, is at least $\epsilon-\frac{8}{3}\cdot q_H^4  \delta^2$.
Hence, we can make the probabilities arbitrarily small while still keeping $\delta$
noticeable (in the order of $q_H^{-2}$). This is important in order to extract the secret key successfully. Specifically, the following
two (seemingly contradictory) conditions have to be fulfilled:
\begin{itemize}
\item We need to ensure that $\qadv$ eventually outputs a valid signature $(r,\com^\star, \ch' ,\resp')$ for some message $m$ for 
the commitment $\com^\star$ of our choice (the one we obtained from the zero-knowledge simulator of the $\Sigma$-protocol which we 
inject into $\O$'s responses). This requires that $\com^\star$ appears with sufficiently large probability in the responses for oracle queries.
\item Secondly, we still require that $\qadv$ has a small probability of distinguishing a true random oracle $\O$
 from the re-programmed one. Otherwise, the adversary may refuse to give a valid signature at all.
\end{itemize}

The following lemma shows that both conditions can be satisfied simultaneously.
\begin{lemma}\label{lem:output}\label{lem:nodistinguish}
Let $\qadv^{\ket{\O}}$ as in Lemma \ref{lem:zhandrycustom}, and let $\O'$ be the oracle obtained by reprogramming $\O$ on a 
fraction $\delta$ of its possible inputs $(\pk,m,r)$ such that $\O'({\pk, m, r}) = (\com^\star, \ch')$ for values $\com^\star$ 
and $\ch'$. Let $m$, $\sigma=(r,\com, \ch ,\resp)$ be the output of $\qadv^{\ket{\O^\prime}}$ on input $\pk$.
Then,
\[
\Pr \left[ \vf^{\,\O^\prime}(\pk,m,\sigma)=1\;\und\;(\com, \ch) = \O'(\pk,m,r)=(\com^\star, \ch') \right]
   \geq \delta \cdot\epsilon- \frac{8}{3}\cdot q_H^4  \delta^2\; .
\]
\end{lemma}

\begin{proof}
Consider the probability that we first run the adversary on the original oracle $\O$ and check if it successfully forges
a signature for message $m$, and then we also verify that its output $(\pk,m,r)$ is thrown to
$(\com^\star, \ch')$ under $\O'$. We claim that
\[
 \Pr \left[ \vf^{\,\O}(\pk,m,\sigma)=1\;\und\;\O'(\pk,m,r)=(\com^\star, \ch') \right]  \ge  \epsilon\cdot\delta.
\]
This follows from the independence of the events: the oracle $\O'$ re-programs the output with probability $\delta$,
independently of $\adv$'s behavior when interacting with oracle $\O$. Next, we argue that
\[  \Pr \left[ \vf^{\,\O^\prime}(\pk,m,\sigma)=1\;\und\;\O'(\pk,m,r)=(\com^\star, \ch') \right]  \ge
\delta \cdot\epsilon- \frac{8}{3}\cdot q_H^4  \delta^2. \]
Note that the difference is now that the adversary interacts with oracle $\O'$, and that we also verify the adversary's
success with respect to $\O'$. Instructively, the reader may imagine that, after the adversary's attack ends, we also check
that $\O'(\pk,m,r)=(\com^\star, \ch')$; the equation $\O'(\pk,m,r)=(\com, \ch)$, as in the lemma's claim,
trivially follows already if verification holds.

According to the previous lemma, switching to oracle $\O^\prime$ can change the distance of the output distribution of $\adv$
when playing against $\O'$ instead of $\O$ (including the final verification) by at most $\frac{8}{3}\cdot q_H^4  \delta^2$.
Hence, since the subsequent computation and check $\O'(\pk,m,r)=(\com^\star, \ch')$ cannot increase this distance, we conclude
that the probability for event
\[ \vf^{\,\O^\prime}(\pk,m,\sigma)=1\;\und\;\O'(\pk,m,r)=(\com^\star, \ch')=(\com,\ch) \]
cannot be smaller than the claimed bound.
\qed
\end{proof}

The previous lemma informally tell us that, in order to succeed, we have to balance between a large $\delta$ to increase the chances of the adversary outputting a signature containing our desired $\com^\star$, and a small $\delta$ to avoid that the adversary detects the reprogrammed oracle.

\mysubsection{Security Proof} \label{app:positiveproof}
We are now ready to prove the main theorem.

\begin{proof}[of Theorem \ref{th:fs}]
We assume towards contradiction the existence of an efficient quantum adversary $\qadv$ which, on input a public key $\pk$, outputs a
valid forgery $(m,\sigma)$ under $\pk$ with non-negligible probability $\epsilon$, hence breaking the existential
unforgeability of the signature scheme. This adversary has access to
a quantum-accessible random oracle $\O$ with 
$\O (\pk,m_i,r_j)= (\com_{i,j},\ch_{i,j})$, and to a signing oracle $\S$ for the key $\sk$ producing, on input a classical 
strings $m$, (classical) signatures $\sigma=(r,\com,\ch,\resp) \exec \sign^H(\sk,m)$.

The adversary $\qadv$ gets $\pk$ as an input, and is then allowed to perform up to $q_H=\poly(\lambda)$ queries to $\O$ in superposition, 
and up to $q_S=\poly(\lambda)$ classical queries to $\S$. Recall that the signer still operates on classical bits.
Then, after running for $\poly(\lambda)$ time, adversary $\qadv$ produces (with probability $\epsilon$) a valid forgery $(m,\sigma)$ under 
$\pk$ such that $m$ has never been asked to the signing oracle $\S$ throughout $\qadv$'s execution (i.e., $m$ is a fresh message).
We assume that $q_H$ also covers a classical query of the verifier to check the signature.
Under these assumptions we show how to build an efficient quantum adversary $\qbdv$, with access to $\qadv$ as a subroutine, and which
is able to break the scheme's underlying hard mathematical problem with non-negligible probability. That is, $\qbdv$ on input $x$ generated according to
$\igen(1^\secpar)$, is able to output a valid witness $w'$ to statement $x$, i.e., $(x,w')\in \cR$. The adversary $\qbdv$ works as follows:
\begin{itemize}
\item On input statement $x$, it first runs a simulation of the underlying $\Sigma$-protocol to obtain a valid transcript $(\com^\star,\ch^\star,\resp^\star)$.
This is possible because of the honest-verifier zero-knowledge property. Note also that this does not require access to the random oracle.
Also note that we assume for simplicity that the oblivious commitment is a random string; else we would need to run
$\Smplrnd$ on $\com^\star$ now to derive $\rho$, and use $\rho$ in the hash reply (and argue that this is indistinguishable).
\item Then, $\qbdv$ simulates an oracle $\O_0$ which is obtained by reprogramming a (simulated) quantum random oracle $\O$ over a fraction $\delta$ 
of its possible inputs $(\pk,m,r)$  with the value $(\com^\star, \ch')$. Here, $\delta$ is some non-negligible probability in the security parameter, and $\ch'$ is a
fix, arbitrarily chosen challenge different from $\ch^\star$. That is, $\O_0 ({\pk,m,r}) = (\com^\star,\ch')$ with probability $\delta$,
and random elsewhere.
\item Next, $\qbdv$ invokes $\qadv$ on input $\pk = x$. 
\item Whenever $\qadv$ performs the $i$-th query to $\S$ for signing a message $m_i$, adversary $\qbdv$ does the following:
	\begin{itemize}
	\item choose a random value $r_i \rand \Rnd$;
	\item execute the honest-verifier zero-knowledge simulator $\Sim$ of the identification scheme, obtaining a valid
	(simulated) transcript $(\com_i, \ch_i, \resp_i)$;
	\item reprogram $\O_{i-1}$ with value $(\com_i, \ch_i)$ for the input $(\pk,m_i,r_i)$. We denote by $\O_i$ the reprogrammed oracle 
	after the $i$-th query to the signing oracle;
	\item then output the signature $\sigma_i=(r_i,\com_i,\ch_i,\resp_i)$ as $\S$'s reply to $\qadv$.
	\end{itemize}
\item Finally, when $\qadv$ outputs a (hopefully valid) forgery $(m, \sigma)$, where $\sigma=(r,\com,\ch,\resp)$, algorithm $\qbdv$
aborts if $\com \neq \com^\star$ or $\ch = \ch^\star$. Otherwise, it uses the special soundness extractor $\extractor$ of the underlying
$\Sigma$-protocol on input $(\com^\star,\ch^\star,\resp^\star)$ and $(\com,\ch,\resp)$ to obtain a valid witness $w'$ for $x$.
\end{itemize}
Note that we can formally let $\qbdv$ implement the hash evaluations by a classical algorithm with access to a
random oracle, basically hardwiring all changes due to re-programming into the code of the algorithm. In a second step
we can eliminate the random oracle, either via quantum-accessible pseudorandom functions \cite{AC:BDFLSZ11}, or without
any assumptions by using $q$-wise independent function as shown in~{\cite[Theorem~6.1]{C:Zhandry12}}. These
functions can be implemented by classical algorithms.

We next show that the success probability of our extraction procedure $\qbdv$ is non-negligible given a successful $\qadv$.
The proof follows the common game-hopping technique where we gradually deprive the adversary a (negligible amount) of its success probability. 
We start with $\game_1$ where the adversary attacks the original scheme.
\begin{description}
\item[$\game_1$.] This is $\qadv$'s original attack on the signature scheme as constructed according to Construction~\ref{constr:sig} 
initialized by public key $\pk$. By assumption we have
\[
\Pr\left[\text{$\qadv$ wins $\game_1$}\right] = \epsilon 
\]
for some non-negligible value $\epsilon$.

\item[$\game_2$.] This game is identical to $\game_1$, except that we abort if $\qadv$ outputs a valid forgery $(m,\sigma)$ where $\sigma$ \emph{does not} contain
the pre-selected commitment $\com^\star$ and challenge $\ch'$. Furthermore, we replace the quantum-accessible random oracle $\O$ with the 
oracle $\O_0$ drawn from a semi-constant distribution $\mysc_\delta$. Recall that $\O_0$ is obtained by reprogramming $\O$ on a fraction 
$\delta$ of its entries with the value $(\com^\star,\ch')$, where $\com^\star$ was obtained by a run of the honest-verifier zero-knowledge 
simulator $\Sim$ on input $x$
and $\ch'$ was picked as in $\qbdv$'s simulation. By Lemma~\ref{lem:output} we have
\[
\Pr\left[\text{$\qadv$ wins $\game_2$}\right] \geq \delta \epsilon - \frac{8}{3} q_H^4 \delta^2\; .
\]

\item[$\game_3^{(1)}$.] As $\game_2$, but this time $\O_0$ is reprogrammed to $\O_1$ (on the single point $(\pk,m_1,r_1)$) as soon as 
$\qadv$ performs its $1^{st}$ classical query $m_1$ to $\S$. From then on, the oracle $\O_1$ always answers consistently with this value.
We need to show that this switching does not change the winning probability
significantly.
For this we basically need to show that, so far, the amplitudes of this value $(\pk,m_1,r_1)$ in the queries to the quantum oracle are 
small, else the adversary may be able to spot some inconsistency.

Let $| \Rnd | = 2^{n} = \exp{(\secpar)}$. We define the value $(\pk,m'_i,r'_j)$ to have \emph{high amplitude} if there exists at least 
one of the quantum queries $\ket{\phi_1},\ket{\phi_2},\ldots$ to the quantum-accessible oracle $\O_0$ \emph{before the signing query},
where the amplitude $\alpha_{i,j}$ associated to the corresponding basis element is such that $|\alpha_{i,j}|^2 \geq 2^{\frac{-{n}}{2}}$. 
Otherwise, the tuple is said to have \emph{low amplitude}. Note that each query to the quantum oracle can have at most
$2^{\frac{{n}}{2}}$ tuples with high amplitude, because the (square of the) amplitudes need to sum up to $1$.

When $\O_0$ is reprogrammed to $\O_1$, the choice of $m_1$ is fixed
(i.e., determined by the $1^{st}$ query of $\qadv$ to $S$), but $r_1$ is still chosen uniformly at random in $\Rnd$.
Since $\qadv$ performs at most $q_H$ queries to the quantum-accessible oracle according to $\O_0$
before the signing query, we have thus at most $q_H\cdot 2^{\frac{{n}}{2}}$
tuples with high amplitude before this query. The probability of hitting such a tuple is then given by:
\begin{equation}\label{eqn:highamp}
\Pr\left[ \text{$(\pk,m_1,r_1)$ has high amplitude} \right] \leq q_H \cdot 2^{\frac{-{n}}{2}}.
\end{equation}
Moreover, provided $(\pk,m_1,r_1)$ has \emph{low} amplitude, and since there are at most $q_H +q_S$ query steps, using Lemma~\ref{th:distances} and Lemma~\ref{th:queryprob} we obtain:
\begin{equation}\label{eqn:BBB}
\left| \qadv^{\ket{\O_0}} - \qadv^{\ket{\O_1}} \right| \leq 4 \sqrt{(q_H + q_S) \cdot 2^{\frac{-{n}}{2}}}.
\end{equation}
Let us assume, on behalf of the adversary, that $\qadv$ fails whenever $(\pk,m_1,r_1)$ has high amplitude. Still, from equations \eqref{eqn:highamp} and \eqref{eqn:BBB}, we have:
\begin{align*}
\Pr\left[\text{$\qadv$ wins $\game_3^{(1)}$}\right] & \geq \Pr\left[\text{$\qadv$ wins $\game_2$}\right] - 4 \sqrt{(q_H + q_S) \cdot 2^{\frac{-{n}}{2}}} - q_H \cdot 2^{\frac{-{n}}{2}} \\
& = \delta \epsilon - \frac{8}{3} q_H^4 \delta^2 - \negl{\secpar}\; .
\end{align*}
Here, we use the fact that reprogramming the oracle for $(\pk,m_1,r_1)$ does not change the adversary's success probability
for a forgery \emph{for a fresh message $m$}. That is, since the adversary's forgery is for $m\neq m_1,m_2,\dots$
it cannot simply copy a signature query as a forgery, but must still forge on the original oracle $\O_0$.
Hence the argument about the winning probability applies as it did for $\O_0$.

We now repeat at most $q_\S$ times the game hopping, from $\game_3^{(1)}$ to $\game_3^{(q_\S)}$, every time repeating the previous game
but switching from $\O_{i-1}$ to $\O_i$ during the $i^{th}$ query to $\S$, each time losing at most a negligible factor in the winning 
probability. Note that the probability of hitting a high amplitude with the signature generation in each hop increases to at most
$q_H \cdot 2^{\frac{-{n}}{2}}+q_S\cdot 2^{-n}$ when taking into account the at most $q_S$ hash queries in the previous
signature requests, but this remains negligible.
After $q_\S$ steps we reach the following game.

\item[$\game_3^{(q_\S)}$.] As $\game_2$, but now $\O_0$ is dynamically reprogrammed as a sequence $\O_1 , \ldots , \O_{q_S}$ throughout 
all of the $\qadv$'s queries to $\S$. We have
\[
\Pr\left[\text{$\qadv$ wins $\game_3^{(q_S)}$}\right] \geq \delta \epsilon - \frac{8}{3} q_H^4 \delta^2- \negl{\secpar}\; .
\]

\item[$\game_4$.] As before, but now $\S$ is just simulated through the zero-knowledge simulator $\Sim$ of the underlying $\Sigma$-protocol. 
If, by contradiction, $\qadv$'s winning probability is affected by more than a negligible amount in so doing, then we could use $\qadv$ to 
build an efficient distinguisher between `real' and `simulated' transcripts of the $\Sigma$-protocol.  This would require a distinguisher
with access to a random oracle, in order to simulate the game. According to \cite[Theorem 6.1]{C:Zhandry12}, however, we can
simulate the oracle via $q$-wise independent functions (which exists without requiring cryptographic assumptions).
Furthermore, a hybrid argument can be used to reduce the case of $q_S$ proofs to a single proof.
\[
\Pr\left[\text{$\qadv$ wins $\game_4$}\right] \geq \delta \epsilon  -\frac{8}{3} q_H^4 \delta^2 - \negl{\secpar}\; .
\]

\item[$\game_5$.] Finally, in this game 
the special soundness extractor $\extractor$ is run on the transcript obtained from $\qadv$'s output from the previous game.
Change the winning condition of $\qadv$ such that the adversary wins if this extraction yields a valid witness $w'$ for $x$.
If the winning probability 
in this game is more than negligibly far from the winning probability of $\qadv$ in the previous game 
then this can only be due to the fact that the simulated proof with $(\com^\star,\ch^\star,\resp^\star)$ cannot be
accepted by the verifier; else the extractor would be be guaranteed to work for this proof and the (accepted) signature.
But this would allow an easy distinguisher against the zero-knowledge property, similar to the previous games. Hence:
\[
\Pr\left[\text{$\qadv$ wins $\game_5$}\right] \geq \delta \epsilon - \frac{8}{3} q_H^4 \delta^2 - \negl{\secpar}\; .
\]
\end{description}
Note that $\qadv$'s winning condition in the final game corresponds exactly to the probability of $\qbdv$
successfully deriving a witness $w'$ for its input $x$. This winning probability can be maximized (by zeroing the first derivative in $\delta$) by choosing:
\[
\delta = {\frac{3\epsilon}{16 q_H^4}}\; .
\]
This yields:
\[
\Pr\left[\text{$\qadv$ wins $\game_5$}\right] \geq \frac{3\epsilon^2}{16 q_H^4} - \negl{\secpar}\;,
\]
which is non-negligible.
This concludes the proof of the main theorem.\qed
\end{proof}

\subsection{Example Instantiation}

In this section, we present an instantiation of $\Sigma$-protocols with \octxt\ which is secure against quantum adversaries. 
We look at the lattice-based signature scheme by Lyubashevsky~\cite{EC:Lyubashevsky12} which is  obtained by applying the FS 
transformation. The security of this signature scheme is reduced to the hardness of the Small Integer Solution (SIS) problem, 
which is believed to be hard even for quantum algorithms.

Similarly, other works using the FS transformation and relying on the quantum hardness of the underlying primitives, are not 
known to be necessarily secure against quantum adversaries, 
e.g., \cite{AC:Lyubashevsky09,EPRINT:BarMis10,AC:GorKatVai10,C:SakShiHiw11,PKC:Sakumoto12,CHES:GnuLyuPop12,EC:AFLT12,SCN:CamNevRuc12,EC:AJLTVW12}. 
This holds also for signature schemes obtained from the FS extension~\cite{AFRICACRYPT:ADVGC12} for multi-pass identification 
protocols (e.g., \cite{PROVSEC:CLRS10,SAC:CayVerAla10,C:SakShiHiw11,PKC:Sakumoto12}). Furthermore, the FS transform can be applied 
on the identification and zero-knowledge protocols \cite{C:MicVad03,PKC:Lyubashevsky08,AC:KawTanXag08,MGS11} as they are secure 
against quantum adversaries. Still, the converted signature scheme is not necessarily quantum-secure anymore. A similar patch 
approach, as we describe for~\cite{EC:Lyubashevsky12}, can also applied to most of the aforementioned schemes.

\myparagraph{Trapdoors to SIS instances}
We are going to illustrate how a $\Sigma$-protocol with \octxt\ can be obtained through our patch. Basically, we need to provide 
the prover with a trapdoor to extract a candidate preimage to a given commitment. In the scheme from~\cite{EC:Lyubashevsky12} the 
prover has to solve an SIS instance. Roughly speaking, the prover has to find preimages for functions 
$f_\mathbf{A}(\textbf{v}):=\mathbf{Av}$ for $\mathbf{A}\rand \ZZ_q^{n \times m}$ where $v$ is distributed according to the discrete 
Gaussian distribution $D_s$ over $\ZZ^m$ with standard deviation~$s$. The parameters $q,n,m$ as well as $s$ determine the hardness 
of the SIS instance.

From~\cite{Ajtai99,STOC:GenPeiVai08,AlwenP09,C:Peikert10,EC:MicPei12} we know the existence of trapdoors allowing to sample such 
preimages. The most efficient construction from~\cite{EC:MicPei12} finds preimages of length $\beta\approx s \sqrt{m}$ for lattice 
dimension $m\approx 2n \log q$ with (at least) $s\approx 16 \sqrt{n \log q}$. Let $T$ denote the trapdoor for function 
$f_\mathbf{A}$ which is generated together with matrix $\mathbf{A}\rand \ZZ_q^{n \times m}$ by algorithm $\mathsf{GenTrap}(1^n,1^m,q)$. 
Then, the function $\mathsf{SampleD}(\mathbf{T},\mathbf{A},\mathbf{X},s)$ samples an element $\mathbf{x}$ from the distribution 
within negligibly close (in $n$) statistical distance of $D^m_s$ such that $\mathbf{Ax}=\mathbf{X}$ (see, e.g.,
Algorithm~3 of \cite{EC:MicPei12}). 

\myparagraph{Our Patch on the ID Scheme within~\cite{EC:Lyubashevsky12}} We take as input the identification (ID) scheme from which 
the signature scheme in~\cite{EC:Lyubashevsky12} is derived from. Now, we provide the prover with necessary trapdoor information in 
order to enable the prover to respond to a challenge for \octxt. The scheme is parameterized by security parameter $n,q,d,k$, and 
$\eta$. Moreover, $m\approx 2n \log q$\footnote{In the original paper, the author sets $m\approx 64+n \log q/\log (2d+1)$. We slightly 
increase $m$ in order to obtain a trapdoor. This merely strengthens the underlying hardness.}, $\kappa$ is chosen such that 
$2^\kappa \cdot \binom{k}{\kappa} \ge 2^{100}$, and $s\approx 12d\kappa\sqrt{m}$.

The prover first runs $(\mathbf{A}\rand \ZZ_q^{n \times m},\mathbf{T})\exec \mathsf{GenTrap}(1^n,1^m,q)$.
The prover's secret is a matrix \linebreak $\mathbf{S}\rand \{-d,\ldots,0,\ldots,d\}^{m \times k}$ and the trapdoor $T$. The corresponding 
public key consists of the matrices $\mathbf{A}$ and $\mathbf{R}=\mathbf{AS}$. The prover \prov\ picks first a random 
string~$r\rand\{0,1\}^\secpar$, and sends it over to the verifier. The verifier \ver\ randomly picks a challenge 
$\mathbf{c}\rand V=\{\mathbf{v}\;:\; \mathbf{v}\in\{-1,0,1\}^k,\norm{\mathbf{v}}_1 \le \kappa\}$ and computes the commitment 
$\mathbf{Y}\exec \mathbf{Ay}$ for random $\mathbf{y}\in\ZZ^m$ sampled according to $D_s^m$. The verifier forwards both $c$ and 
$\mathbf{Y}$  to \prov.
Now, \prov\ samples first a valid preimage $\mathbf{y}'$ of $\mathbf{Y}$ under function~$f_\mathbf{A}$ through algorithm 
$\mathsf{SampleD}$, and then computes $\mathbf{z}\leftarrow \mathbf{Sc+y'}$. With a certain probability $\rho$ (depending on the 
public parameters, $\mathbf{z}$, and $\mathbf{Sc}$) the pair $(\mathbf{z},\mathbf{c})$ is handed over to \ver. Upon receiving 
$\mathbf{z}$, \ver\ accepts iff $\norm{\mathbf{z}}\le \eta s \sqrt{m}$ and $\mathbf{Rc}=\mathbf{Y}-\mathbf{Az}$.
The underlying interactive scheme is also given in Figure~\ref{fig:ci-id}.
Note that we can assume that $\rho$ is sufficiently large such that we simply let the signer occasionally fail; to get a valid signature
repeatedly call the signer about the same message $m$.

\begin{figure}[t]
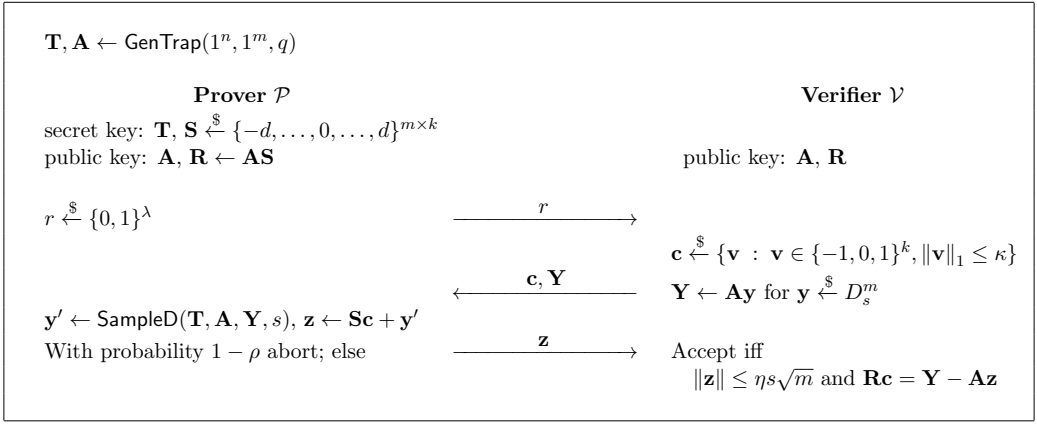

\begin{center}
{
\scalebox{0.75}{
\begin{tabular}{|@{\hspace*{0.75em}}ll@{\hspace*{1.5em}}c@{\hspace*{1.5em}}l@{\hspace*{0.75em}}|}
\hline
& & & \\
& $\mathbf{T},\mathbf{A}\leftarrow\mathsf{GenTrap}(1^n,1^m,q)$ & & \\
& & & \\
&\multicolumn{1}{c}{\textbf{Prover} $\prov$} & & \multicolumn{1}{c|}{\textbf{Verifier} $\ver$} \\
& \multicolumn{2}{l}{secret key: $\mathbf{T}$,  $\mathbf{S}\rand \{-d,\ldots,0,\ldots,d\}^{m \times k}$}   &  \\
& \multicolumn{1}{l}{public key: $\mathbf{A}$, $\mathbf{R}\leftarrow\mathbf{AS}$}  &  & \multicolumn{1}{l|}{public key: $\mathbf{A}$, $\mathbf{R}$} \\
& & & \\
& $r\rand \{0,1\}^\secpar$& $\xrightarrow{\makebox[8em]{$r$}}$& \\
& & & $\mathbf{c} \rand \{\mathbf{v}\;:\; \mathbf{v}\in\{-1,0,1\}^k,\norm{\mathbf{v}}_1 \le \kappa\}$ \\
& & $\xleftarrow{\makebox[8em]{$\mathbf{c},\mathbf{Y}$}}$ &  $\mathbf{Y}\exec \mathbf{Ay} \text{ for } \mathbf{y}\rand D^m_s$ \\
& $\mathbf{y'} \leftarrow \mathsf{SampleD}(\mathbf{T},\mathbf{A},\mathbf{Y},s)$, $\mathbf{z} \leftarrow \mathbf{Sc}+\mathbf{y'}$ &&\\
& With probability $1-\rho$ abort; else & $\xrightarrow{\makebox[8em]{$\mathbf{z}$}}$& Accept iff \\ & & & \quad $\norm{\mathbf{z}}\le \eta s \sqrt{m}$ and $\mathbf{Rc}=\mathbf{Y}-\mathbf{Az}$\\
&& & \\ \hline
\end{tabular}
}
}
\end{center}
\caption{\small Patched $\Sigma$-Protocol from~\cite{EC:Lyubashevsky12}}
\label{fig:ci-id}
\end{figure}

\myparagraph{On the Quantum Security}
We stress that the resulting (identification) scheme has now \octxt, i.e., it satisfies Definition~\ref{def:vcc}.

Note that the security as an identification scheme does not depend on the first message sent by the prover. However, if one 
converts the ID protocol to a signature scheme, this message serves as the randomness input to the hash function together 
with the message to be signed. As such, a signature on a message $m$ consists of (randomized) $\sigma=(r,\resp)$ where 
$(\com,\ch)=(\mathbf{Y},\mathbf{c})\exec H(\pk,r,m)$.
Hence, the signature scheme obtained by the FS transformation (see Construction~\ref{constr:sig}) on the above identification 
scheme is secure against quantum adversaries in the quantum-accessible random-oracle model following from the result of 
Theorem~\ref{th:fs}.

Notice that our resulting signature scheme is close to a variant of the hash-and-sign  signature scheme (GPV) by Gentry, 
Peikert, and Vaikuntanathan~\cite{STOC:GenPeiVai08}. The GPV signature scheme uses a pre-images sampleable trapdoor function 
(PSF), and signing a message here is basically providing a preimage to the hashed message. If in our construction, the 
commitment and challenge is merely the hash of the message, both signatures coincide. In a concurrent work~\cite{EPRINT:BonZha13}, 
the GPV signature scheme is proven secure in the QROM. Interestingly, this gives us two different proof approaches for similar 
schemes. While Boneh and Zhandry~\cite{EPRINT:BonZha13} give security results for the hash-and-sign paradigm and thereby show 
the security of the GPV signature, our security proof follows immediately from Theorem~\ref{th:fs} once the scheme by Lyubashevsky 
is patched to have oblivious commitments. Note that, ideally, our result applies to future FS schemes with improved
efficiency as well.

\section{Conclusion}\label{sec:conclusion}

Our impossibility result indicates that the Fiat-Shamir paradigm should be
taken with great caution when used to argue quantum resistance.
A proof for a scheme in the classical random-oracle model, even if the underlying
problem is quantum-resistant, may not yield a protocol which
is also secure in the QROM. For some schemes, however, a formal proof in
the quantum random oracle is possible after a minor modification. Interestingly,
this modification may first weaken the scheme, e.g., remove active security.

It remains open to bypass our black-box separation result by other means, e.g.,
by using witness-\emph{dependent} commitments, to extend the class of admissible
protocols for which the transformation yields a secure scheme in the QROM. Alternatively,
one may try to give a ``Fiat-Shamir-like'' transformation which also yields
secure signature schemes in the QROM. Natural candidates would be the constructions
with online extractors by Pass~\cite{C:Pass03} and by Fischlin~\cite{C:Fischlin05}, which
potentially also circumvent the rewinding problem.
We leave this as an interesting open question.

\section*{Acknowledgments}

We thank the anonymous reviewers for some valuable comments. We also thank Dominique
Unruh for useful discussions on black-box extractors. Marc Fischlin is supported by the
Heisenberg grant Fi 940/3-1 of the German Research Foundation (DFG). Tommaso Gagliardoni
is supported by the German Federal Ministry of Education and Research (BMBF) within
EC-SPRIDE. This work was also supported by CASED (\texttt{www.cased.de}).

\newcommand{\etalchar}[1]{$^{#1}$}
\def\shortbib{0}\providecommand{\modbibauthor}[1]{#1}
  \providecommand{\modbibtitle}[1]{\textit{#1}}
  \providecommand{\modbibbtitle}[1]{\textit{#1}}
  \providecommand{\modbibjournal}[1]{#1} \providecommand{\modbibyear}[1]{#1}
  \providecommand{\modbibvolume}[1]{\normalfont #1}
  \providecommand{\modbibseries}[1]{\normalfont #1}
  \providecommand{\modbibtype}[1]{#1} \providecommand{\modbibpages}[1]{#1}
  \providecommand{\modbibbooktitle}[1]{\normalfont #1}
  \providecommand{\modbibpublisher}[1]{#1}
  \providecommand{\modbibeditor}[1]{#1}
  \providecommand{\modbibeditortxt}[1]{(ed.)}
  \providecommand{\modbibvolumetxt}{Volume} \providecommand{\modbibintxt}{}
  \providecommand{\modbiboftxt}{ of } \providecommand{\modbibpagetxt}{page}
  \providecommand{\modbibpagestxt}{pages}
  \providecommand{\modbibnumbertxt}{number}
  \providecommand{\modbibcombibitem}[2]{\bibitem[#1]{#2}}
  \providecommand{\modbibcomvolume}[1]{\modbibvolumetxt{} #1}
  \providecommand{\modbibcompage}[1]{\modbibpagetxt{} #1}
  \providecommand{\modbibcompages}[1]{\modbibpagestxt{} #1}
  \providecommand{\modbibcomthebibliography}[1]{\begin{thebibliography}{#1}}

\end{document}